\documentclass{article}
\usepackage[utf8]{inputenc}
\usepackage{natbib}
\usepackage{amsthm, amsmath, amssymb}
\usepackage[hidelinks]{hyperref}
\usepackage{tikz}
\usepackage{comment}

\title{On Selecting and Conditioning in Multiple Testing and Selective Inference}
\author{Jelle Goeman\footnote{Leiden University Medical Center, Department of Biomedical Data Sciences.} \and Aldo Solari\footnote{University of Milano-Bicocca, Department of Economics, Management and Statistics.}\,\,\footnote{Ca' Foscari University of Venice, Department of Economics.}}
\date{July 2022}
\usepackage[margin=3cm]{geometry}
\usepackage{times}

\newtheorem{proposition}{Proposition}
\newtheorem{observation}{Observation}

\newtheorem{lemma}{Lemma}

\begin{document}

\maketitle

\begin{abstract}
    We investigate a class of methods for selective inference that condition on a selection event. Such methods follow a two-stage process. First, a data-driven (sub)collection of hypotheses is chosen from some large universe of hypotheses. Subsequently, inference takes place within this data-driven collection, conditioned on the information that was used for the selection. Examples of such methods include basic data splitting, as well as modern data carving methods and post-selection inference methods for lasso coefficients based on the polyhedral lemma. In this paper, we adopt a holistic view on such methods, considering the selection, conditioning, and final error control steps together as a single method. From this perspective, we demonstrate that multiple testing methods defined directly on the full universe of hypotheses are always at least as powerful as selective inference methods based on selection and conditioning. This result holds true even when the universe is potentially infinite and only implicitly defined, such as in the case of data splitting. We provide a comprehensive theoretical framework, along with insights, and delve into several case studies to illustrate instances where a shift to a non-selective or unconditional perspective can yield a power gain.
\end{abstract}

\section{Introduction}

When many potential research questions are considered simultaneously, researchers often only report a subset of the findings, typically the most striking, interesting, or surprising ones. When interpreting results  selected in this way, it is crucial to recognize that the evidence for the findings may be exaggerated due to the selection process. The field of selective inference, also known as multiple testing, strives to adjust inference for this data-driven selection of research questions. Selective inference methods ensure that the number or proportion of incorrect findings among the final reported findings remains small. The selective inference literature is large and well-established
\citep{benjamini2010, dickhaus2014, taylor2015, taylor2018survey,  benjamini2019, cui2021,
kuchibhotla2021,
zhang2022}. Classic approaches in the field either control of the familywise error rate or the false discovery rate.

Recently, a two-step approach to selective inference has gained popularity \citep{fithian2014,  lee2016, tibshirani2016, 
charkhi2018, 
bi2020}. In this conditional approach, the data are first used to select a small set of hypotheses of interest from a large universe of hypotheses. Next, inference is conducted on the selected hypotheses using the same data, but conditional on the information used for the selection. The conditional approach can be seen as a sophisticated generalization of data splitting. In data splitting, part of the subjects are used to select hypotheses, and the rest for inference on them. Conditional approaches similarly use part of the information in the data for selection, and the remainder for inference.
Proponentes of conditional selective inference often contrast their approach to classical methods, suggesting that the conditional way of thinking represents the most fitting philosophy for selective inference, addressing the problem of selection in the most effective way. E.g., \cite{kuffner2018} state ``The appropriate conceptual framework for valid inference is that discussed in the statistical literature as ‘post-selection inference’, which 
[...] requires conditioning on the selection event and control of the error rate of the inference given it was
actually performed.''.


Conditional selective inference methods return a selection-adjusted $p$-value for each of the selected hypotheses, or a selection-adjusted confidence interval for each of the selected parameters. The key property of these selection-adjusted measures, i.e., uniformity under the null for $p$-values, and coverage for confidence intervals, holds conditional on the selection event. In the situation that more than one such $p$-value or confidence interval is returned, some authors argue for a further round of adjustment for multiple testing \citep[e.g.][]{benjaminiy2019}, while others consider it as an option \citep[e.g.][]{hyun2021} or do not perform any further correction \citep[e.g.][]{lee2016}. Even when further multiple testing is done, however, this is generally not considered part of the conditional selective inference method itself, but simply a post-processing of the selection-adjusted $p$-values or confidence intervals returned by the method. This detachment of the selection and inference steps has been criticized as circular, because the interpretation of selected but not significant hypotheses is not always clear\citep[Section B.1, supplement]{weinstein2020}.

In this paper we adopt an alternative, holistic perspective on conditional selective inference. We argue that any follow-up, in terms of multiple testing or lack thereof, on the selection-adjusted $p$-values should be regarded as an integral component of the selective inference method. From this point of view, conceptual differences between conditional selective inference and classical methods largely vanish. 
We argue that for every conditional selective inference method, there exists a method that is not selective and not conditional that always rejects all the hypotheses the original method rejects, and possibly more. We give several general conditions under which unconditional and non-selective methods are truly superior to selective conditional methods, and present several worked-out examples. Our results hold for methods returning selection-adjusted $p$-values or selection-adjusted confidence intervals, and apply to a variety of error rates.

\section{Conditional selective inference: basics} \label{sec overview}

Let $\mathrm{P} \in M$ be a probability measure, where $M$, the model, is a collection of probability measures defined on a common outcome space $\Omega$. We will first focus on hypothesis testing, addressing confidence intervals in Section \ref{sec FCR}. A hypothesis is a subset $H \subseteq M$, and $H$ is true if $\mathrm{P} \in H$, and false otherwise. We have data $X$, distributed according to $\mathrm{P}$.  

Conditional selective inference procedures consider a random collection of hypotheses. Sometimes we assume that we know the distribution of $S$, such as, for example, when $S$ consists of the null hypotheses corresponding to the active set of a lasso regression. In other cases we may have only a realization of $S$ without knowledge of its distribution, such as when $S$ was chosen freely by a user on the basis of the first half of the data. In both cases, however, we will assume that we know what part of the information in $X$ was used to select $S$. In the lasso example we know this information because we know how $S$ was calculated. In the data-splitting example we know that the user only saw part of the data. 

We will illustrate our general discussion with a recurring toy example. Assume that two $p$-values $P_1, P_2$ are independent, and that $P_1 \sim \mathcal{U}(0,1)$ under hypothesis $H_1$ and $P_2 \sim \mathcal{U}(0,1)$ under $H_2$. A simple selective inference procedure could discard hypotheses for which the $p$-values are larger than some fixed $\lambda$. In this case we have $S = \{i\colon P_i \leq \lambda\}$. This a situation considered by \cite{zhao2019} and \cite{ellis2020}. A similar selection set would arise when doing inference based on the polyhedral lemma if the design is orthogonal \citep{reid2017}. 

The collection $S$ is drawn from a larger universe of hypotheses, which often remains implicit in the selective inference literature. Let $\mathcal{S} = \{S(\omega) \colon \omega \in \Omega\}$ be the collection of all possible realizations of $S$. We define the universe $U$ as all hypotheses that could have been in $S$. Formally,
\[
U = \bigcup_{\omega \in \Omega} S(\omega). 
\]
Unlike $S$, the universe $U$ is fixed. It can be huge, or even infinite. For example, when $S$ are null hypotheses for the regression coefficients of the active set of a lasso regression, then $U$ contains all null hypotheses for all regression coefficients for all covariates adjusted for all possible sets of other covariates \citep[compare][]{berk2013,   bachoc2020}. In other cases $U$ is even unknown. For example, if $S$ was chosen freely by the user using half of the data, then $U$ contains all hypotheses the user would have chosen if the data would have been different. In this case we know nothing about $U$ except that it is a superset of $S$. To avoid trivial problems, we assume that $U \neq \emptyset$. In the toy example we have $U = \{1,2\}$.

Conditional selective inference methods define selection-adjusted  $p$-values $p_{H|S}$ for $H \in S$. These have the property that, for every $\alpha \in [0,1]$, 
\begin{equation} \label{eq adjusted p}
\sup_{\mathrm{P} \in H} \mathrm{P} (p_{H|S} \leq \alpha \mid S) \leq \alpha. 
\end{equation}
The selection-adjusted $p$-value differs the usual definition of the $p$-value $p_H$, i.e., for every $\alpha \in [0,1]$, $\sup_{\mathrm{P} \in H} \mathrm{P} (p_{H} \leq \alpha) \leq \alpha,$ because it conditions on $S$. 
By conditioning on the selection event $S$, the selection-adjusted $p$-value discards the information used for that selection. It uses as evidence against the selected hypothesis $H$ only the remainder of the information in the data. Conditioning thus provides a neat separation between the information used for selecting $S$ and for inferring on the hypotheses in $S$. Condition (\ref{eq adjusted p}) remains valid if we condition on more than just $S$, but \cite{fithian2014} argued that it is optimal to condition on the minimal amount of information under which $S$ is measurable.

There are many methods for calculating selection-adjusted $p$-values. The most straightforward way to achieve (\ref{eq adjusted p}) is to separate the data into two independent components, writing $X = (X', X'')$ and making sure that $S$ is a function of $X'$ only, while $p_{H|S}$, for every $H \in S$, involves $X''$ only. This is the basic idea of data splitting \citep{moran1973, cox1975, rubin2006, dahl2008,  wasserman2009, rinaldo2019}. 
More sophisticated methods may use the data more efficiently by employing external randomization \citep{tian2018r, rasines2021, leiner2021, panigrahi2022approximate, panigrahi2022exact, dharamshi2023generalized} or multiple data splits \citep{ meinshausen2009, diciccio2020, schultheiss2021}. Some methods split the data adaptively, unmasking the data bit by bit until the user is ready to select the final set $S$ and calculate the $p$-values conditional on that final $S$ \citep{lei2018, duan2020}. If an obvious split of the data is not available, the mathematics of the conditioning can become quite complex. The polyhedral lemma \citep{lee2016, tibshirani2016}, an important breakthrough, provides machinery to condition on selected sets arising in linear regression contexts, such as active sets from lasso regression. This result has been extended and applied in many contexts
\citep{lee2014, 
yang2016, 
tian2017, 
liu2018, 
hyun2018, 
taylor2018, 
heller2019, 
panigrahi2021, 
garcia2022, 
zhao2022selective}.

In the toy example, we can calculate section-adjusted $p$-values by looking at the conditional distribution of the $p$-values under the null. If $i \in S$, we obtain $P_{i|S} = P_i/\lambda$. We will slightly abuse notation throughout the paper, writing $i\in S$ instead of $H_i \in S$ and $P_{i|S}$ for $P_{H_i|S}$; this should cause no confusion. To adjust for the selection, the $p$-value has been multiplied by a factor $1/\lambda$. It is easy to verify that, whenever $i \in S$, we have \[
\mathrm{P}(P_{i|S} \leq t\mid S) = \mathrm{P}(P_i/\lambda \leq t \mid P_i \leq \lambda) = \lambda t/\lambda = t,
\]
so that $P_{i|S}$ fulfils (\ref{eq adjusted p}).

\section{Multiple testing adjustment of selection-adjusted $p$-values}

Having calculated selection-adjusted $p$-values, the usual next step is to decide which of the hypotheses in $S$ can be rejected. A method must be decided for this, be it simply to reject all hypotheses with $p_{H|S}\leq \alpha$ for some $\alpha$, or some more sophisticated multiple testing procedure. Whatever method was chosen, the end result is a random set $R \subseteq S$ of rejected hypotheses.

There are different views on the properties the set $R$ should have, but generally the focus is on avoiding false discoveries. Let 
\[
T_{\mathrm{P}} = \{H \in U\colon \mathrm{P} \in H\}
\]
be the collection of all true hypotheses in $U$. Rejection of $R$ induces $|R \cap T_{\mathrm{P}}|$ false discoveries, giving a false discovery proportion of \[
f_{\mathrm{P}}(R) = \frac{|R \cap T_\mathrm{P}|}{|R| \vee 1}.
\]
To keep false discoveries in check we can control the expectation of some error rate $e_\mathrm{P}(R)$, for which there are many choices \citep{benjamini2010, benjamini2019}, e.g., $e_\mathrm{P}(R) = f_\mathrm{P}(R)$ to control FDR; $e_\mathrm{P}(R) = 1_{f_\mathrm{P}(R)>0}$ to control FWER; or 
$e_\mathrm{P}(R) = 1_{f_\mathrm{P}(R)>\gamma}$ to control FDX-$\gamma$. We assume that $0 \leq e_\mathrm{P}(R) \leq 1$, and that $e_\mathrm{P}(R) =0$ whenever $R \cap T_\mathrm{P} = \emptyset$.

To control a chosen error rate, we bound its expectation by $\alpha$. There are two flavors here. We can control the error rate conditional on $S$, requiring that, for every $\mathrm{P} \in M$ and every $S \in \mathcal{S}$, 
\[
\mathrm{E}_\mathrm{P} [e_\mathrm{P}(R) \mid S ] \leq \alpha,
\]
where $\mathrm{E}_\mathrm{P}(\cdot) = \int_\Omega \cdot\, d\mathrm{P}$ is the expectation corresponding to $\mathrm{P}$. Alternatively, we can aim for unconditional control, requiring that, for every $\mathrm{P} \in M$, 
\[
\mathrm{E}_\mathrm{P} [ e_\mathrm{P}(R) ] \leq \alpha.
\]
Most authors in conditional selective inference advocate control of the conditional error rate \citep{
fithian2014, lee2016, kuffner2018}, though it has been shown that conditioning can sometimes be problematic \citep{kivaranovic2020, kivaranovic2021}.  Other authors, however, have argued for the unconditional error rate, sometimes finding that it leads to more power \citep{wu2010, andrews2019, andrews2022}. Indeed, the conditional error rate is the more stringent one, since conditional control implies unconditional control.

In the toy example, multiple testing is an issue only if $S = \{1,2\}$. If we choose to control FWER at level $\alpha$, we may use the methods of \citet{hochberg1988} or \citet{hommel1988}, which are equivalent in the case of two hypotheses. This method rejects each $H_i$ if $P_{i|S} \leq \alpha/2$, and rejects both hypotheses if $P_{1|S}$ and $P_{2|S}$ are both at most $\alpha$. The resulting procedure is displayed graphically on the left-hand side of Figure \ref{fig example1-1}. Alternatively, we may choose to control FDR. With two hypotheses, the procedure of \cite{benjamini1995} is equivalent to the Hommel/Hochberg-procedure just described, and controls FWER as well as FDR. For controlling FDR we can do uniformly better with the minimally adaptive Benjamini-Hochberg procedure \cite[MABH,\ ][]{solari2017}. In the case of two hypotheses, this procedure also uniformly improves the adaptive procedure of \cite{Benjamini2006}. MABH rejects each $H_i$ if $P_{i|S} \leq \alpha/2$; it rejects both hypotheses if either $P_{1|S}$ and $P_{2|S}$ are both at most $\alpha$, or if the smallest is at most $\alpha/2$ and the largest at most $2\alpha$. It is displayed graphically in the middle part of Figure \ref{fig example1-1}.

\begin{figure}[!ht]
\centering
\def\lam{.7}
\def\restlam{.3}
\def\alp{.3}

\begin{tikzpicture}[scale=3.8]

\draw (0,0) rectangle (1,1);

\filldraw[gray!50] (0, \lam) rectangle (\lam*\alp, 1);
\filldraw[gray!50] (\lam, 0) rectangle (1,\lam*\alp) ;
\filldraw[gray!50] (\alp*\lam, 0) rectangle (\lam,\lam*\alp/2);
\filldraw[gray!50] (0,\alp*\lam) rectangle (\lam*\alp/2,\lam);
\filldraw (0,0) rectangle (\lam*\alp,\lam*\alp);
\draw (1,\lam) -- (1.01,\lam) node[right]{$\lambda$};
\draw (1,\lam*\alp) -- (1.01,\lam*\alp) node[right]{$\lambda\alpha$};
\draw (1,\lam*\alp/2) -- (1.01,\lam*\alp/2) node[right]{$\lambda\alpha/2$};

\draw (0,0) -- (1,0) node[midway, below]{$P_1$};
\draw (0,0) -- (0,1) node[midway, left]{$P_2$};
\draw (1, 0) -- (1, 1);
\draw (0,1) -- (1, 1) node[midway, above, yshift=.2cm]{(FWER)};
\draw (0,0) node[below left]{0};
\draw (0,1) node[left]{1};
\draw (1, 0) node[below]{1};

\draw (\lam, 0) -- (\lam, 1) node[below left]{$\{1\}$};
\draw (0,  \lam) -- (1, \lam) node[below left]{$\{2\}$};
\draw (1, 1) node[below left]{$\emptyset$};
\draw (\lam, \lam) node[below left]{$\{1,2\}$};

\begin{scope}[xshift = 1.4 cm]
\draw (0,0) rectangle (1,1);

\filldraw[gray!50] (0, \lam) rectangle (\lam*\alp, 1);
\filldraw[gray!50] (\lam, 0) rectangle (1,\lam*\alp) ;
\filldraw[gray!50] (\alp*\lam, 0) rectangle (\lam,\lam*\alp/2);
\filldraw[gray!50] (0,\alp*\lam) rectangle (\lam*\alp/2,\lam);
\filldraw (0,0) rectangle (\lam*\alp,\lam*\alp);
\filldraw (0,0) rectangle (\lam*\alp/2,2*\lam*\alp);
\filldraw (0,0) rectangle (2*\lam*\alp,\lam*\alp/2);
\draw (1,\lam) -- (1.01,\lam) node[right]{$\lambda$};
\draw (1,\lam*\alp) -- (1.01,\lam*\alp) node[right]{$\lambda\alpha$};
\draw (1,\lam*\alp/2) -- (1.01,\lam*\alp/2) node[right]{$\lambda\alpha/2$};
\draw (1,2*\lam*\alp) -- (1.01,2*\lam*\alp) node[right]{$2\lambda\alpha$};

\draw (0,0) -- (1,0) node[midway, below]{$P_1$};
\draw (0,0) -- (0,1) node[midway, left]{$P_2$};
\draw (1, 0) -- (1, 1);
\draw (0,1) -- (1, 1) node[midway, above, yshift=.2cm]{(FDR)};
\draw (0,0) node[below left]{0};
\draw (0,1) node[left]{1};
\draw (1, 0) node[below]{1};

\draw (\lam, 0) -- (\lam, 1) node[below left]{$\{1\}$};
\draw (0,  \lam) -- (1, \lam) node[below left]{$\{2\}$};
\draw (1, 1) node[below left]{$\emptyset$};
\draw (\lam, \lam) node[below left]{$\{1,2\}$};
\end{scope}

\begin{scope}[xshift = 2.8 cm]
\draw (0,0) rectangle (1,1);

\filldraw[gray!50] (0, \lam) rectangle (\lam*\alp, 1);
\filldraw[gray!50] (\lam, 0) rectangle (1,\lam*\alp) ;
\filldraw[gray!50] (\alp*\lam, 0) rectangle (\lam,\lam*\alp);
\filldraw[gray!50] (0,\alp*\lam) rectangle (\lam*\alp,\lam);
\filldraw (0,0) rectangle (\lam*\alp,\lam*\alp);
\draw (1,\lam) -- (1.01,\lam) node[right]{$\lambda$};
\draw (1,\lam*\alp) -- (1.01,\lam*\alp) node[right]{$\lambda\alpha$};

\draw (0,0) -- (1,0) node[midway, below]{$P_1$};
\draw (0,0) -- (0,1) node[midway, left]{$P_2$};
\draw (1, 0) -- (1, 1);
\draw (0,1) -- (1, 1) node[midway, above, yshift=.2cm]{(FCR)};
\draw (0,0) node[below left]{0};
\draw (0,1) node[left]{1};
\draw (1, 0) node[below]{1};

\draw (\lam, 0) -- (\lam, 1) node[below left]{$\{1\}$};
\draw (0,  \lam) -- (1, \lam) node[below left]{$\{2\}$};
\draw (1, 1) node[below left]{$\emptyset$};
\draw (\lam, \lam) node[below left]{$\{1,2\}$};
\end{scope}

\end{tikzpicture}
\caption{A simple conditional selective inference procedure for two hypotheses inspired by \cite{zhao2019} and \cite{ellis2020}. The left-hand procedure controls FWER, the middle one FDR, and the right-hand side the FCR-inspired error rate (\ref{eq FCR}). The sets displayed in the upper right corner of each quadrant are the realisation of $S$ in that quadrant. Grey indicates areas in which one hypothesis is rejected; black indicates areas in which both are rejected. The plot uses $\lambda=0.7$ and $\alpha=0.3$.}
\label{fig example1-1}
\end{figure}

So far we have assumed that the error rate only depends on $R$, but not on $S$. This assumption excludes the rate 
\begin{equation} \label{eq FCR}
e_\mathrm{P}(R, S) = \frac{|R \cap T_\mathrm{P}|}{|S| \vee 1},
\end{equation}
that is implied by inference based on confidence intervals controlling the False Coverage Rate \citep[FCR,~][]{benjamini2005}. This is also the rate that is controlled if we do no further multiple testing adjustment on the selection-adjusted $p$-values, but simply reject $R = \{i: P_{i|S} \leq \alpha\}$. This procedure is given on the right-hand side of Figure \ref{fig example1-1}. In the next few sections we will assume that the error rate is a function of $R$ only, but we return to $S$-dependent error rates in Section \ref{sec FCR testing}.

\section{A holistic perspective and main observation} \label{sec holistic}

The approaches described in the previous two sections can be seen as two-stage methods. First, from a universe $U$ of hypotheses a selection $S \subseteq U$ is made. Next, within that selection some hypotheses are rejected, while others are not, and we return $R \subseteq S$. The set $R$ is the final result of any method; it is the set we make inferential claims about. 

Rather than analyzing the two steps $U \to S$ and $S \to R$ separately, in this paper we will take a holistic perspective, viewing the two steps together as a single method $U \to S \to R$, or briefly $U \to R$. By viewing the two steps together we stress that the selection step $U \to S$ and the rejection step $S \to R$ are in the hands of the same analyst. The analyst chooses a method for the selection step $U \to S$ and a method for the inference step $S \to R$. The analyst also chooses what part of the information in the data to spend for the selection step, and what part of the data to reserve for the inference step. 

In the holistic perspective, the choice of $S$, in a procedure $U \to S \to R$, is, therefore, part of the method, and this part may be optimized. The holistic perspective implies that such optimization should be focused on obtaining a larger or more useful set $R$, since $R$, not $S$, represents the final inference of the method. In general, we would like to have as many rejections as possible, while keeping the chosen error rate under control. Moreover, from the holistic perspective all rejections of hypotheses in $U$ are welcome, since every hypothesis in $U$ could have been in $S$.

In the toy example, we can visualize the holistic view of the three procedures simply by removing all reference to $S$ in Figure \ref{fig example1-1}, as shown in Figure \ref{fig example1-holistic}. This now displays three single-step procedures, defined directly on the universe $U = \{1,2\}$, and based on the non-selection-unadjusted $P_1$ and $P_2$. The rejected sets $R$ for the procedures in Figure \ref{fig example1-holistic} are trivially identical to those of their counterparts in Figure \ref{fig example1-1}. However, in the holistic perspective of Figure \ref{fig example1-holistic}, the $\lambda$ that previously determined $S$ now becomes a tuning parameter, freely to be chosen by the analyst before seeing the data. The holistic perspective de-emphasizes the importance of $S$.

\begin{figure}[!ht]
\centering
\def\lam{.7}
\def\restlam{.3}
\def\alp{.3}

\begin{tikzpicture}[scale=3.8]

\draw (0,0) rectangle (1,1);

\filldraw[gray!50] (0, \lam) rectangle (\lam*\alp, 1);
\filldraw[gray!50] (\lam, 0) rectangle (1,\lam*\alp) ;
\filldraw[gray!50] (\alp*\lam, 0) rectangle (\lam,\lam*\alp/2);
\filldraw[gray!50] (0,\alp*\lam) rectangle (\lam*\alp/2,\lam);
\filldraw (0,0) rectangle (\lam*\alp,\lam*\alp);
\draw (1,\lam) -- (1.01,\lam) node[right]{$\lambda$};
\draw (1,\lam*\alp) -- (1.01,\lam*\alp) node[right]{$\lambda\alpha$};
\draw (1,\lam*\alp/2) -- (1.01,\lam*\alp/2) node[right]{$\lambda\alpha/2$};

\draw (0,0) -- (1,0) node[midway, below]{$P_1$};
\draw (0,0) -- (0,1) node[midway, left]{$P_2$};
\draw (1, 0) -- (1, 1);
\draw (0,1) -- (1, 1) node[midway, above, yshift=.2cm]{(FWER)};
\draw (0,0) node[below left]{0};
\draw (0,1) node[left]{1};
\draw (1, 0) node[below]{1};

\begin{scope}[xshift = 1.4 cm]
\draw (0,0) rectangle (1,1);

\filldraw[gray!50] (0, \lam) rectangle (\lam*\alp, 1);
\filldraw[gray!50] (\lam, 0) rectangle (1,\lam*\alp) ;
\filldraw[gray!50] (\alp*\lam, 0) rectangle (\lam,\lam*\alp/2);
\filldraw[gray!50] (0,\alp*\lam) rectangle (\lam*\alp/2,\lam);
\filldraw (0,0) rectangle (\lam*\alp,\lam*\alp);
\filldraw (0,0) rectangle (\lam*\alp/2,2*\lam*\alp);
\filldraw (0,0) rectangle (2*\lam*\alp,\lam*\alp/2);
\draw (1,\lam) -- (1.01,\lam) node[right]{$\lambda$};
\draw (1,\lam*\alp) -- (1.01,\lam*\alp) node[right]{$\lambda\alpha$};
\draw (1,\lam*\alp/2) -- (1.01,\lam*\alp/2) node[right]{$\lambda\alpha/2$};
\draw (1,2*\lam*\alp) -- (1.01,2*\lam*\alp) node[right]{$2\lambda\alpha$};

\draw (0,0) -- (1,0) node[midway, below]{$P_1$};
\draw (0,0) -- (0,1) node[midway, left]{$P_2$};
\draw (1, 0) -- (1, 1);
\draw (0,1) -- (1, 1) node[midway, above, yshift=.2cm]{(FDR)};
\draw (0,0) node[below left]{0};
\draw (0,1) node[left]{1};
\draw (1, 0) node[below]{1};

\end{scope}

\begin{scope}[xshift = 2.8 cm]
\draw (0,0) rectangle (1,1);

\filldraw[gray!50] (0, \lam) rectangle (\lam*\alp, 1);
\filldraw[gray!50] (\lam, 0) rectangle (1,\lam*\alp) ;
\filldraw[gray!50] (\alp*\lam, 0) rectangle (\lam,\lam*\alp);
\filldraw[gray!50] (0,\alp*\lam) rectangle (\lam*\alp,\lam);
\filldraw (0,0) rectangle (\lam*\alp,\lam*\alp);
\draw (1,\lam*\alp) -- (1.01,\lam*\alp) node[right]{$\lambda\alpha$};

\draw (0,0) -- (1,0) node[midway, below]{$P_1$};
\draw (0,0) -- (0,1) node[midway, left]{$P_2$};
\draw (1, 0) -- (1, 1);
\draw (0,1) -- (1, 1) node[midway, above, yshift=.2cm]{(FCR)};
\draw (0,0) node[below left]{0};
\draw (0,1) node[left]{1};
\draw (1, 0) node[below]{1};

\end{scope}

\end{tikzpicture}
\caption{Holistic perspective on the procedures of Figure \ref{fig example1-1}. Grey indicates areas in which one hypothesis is rejected; black indicates areas in which both are rejected.}
\label{fig example1-holistic}
\end{figure}
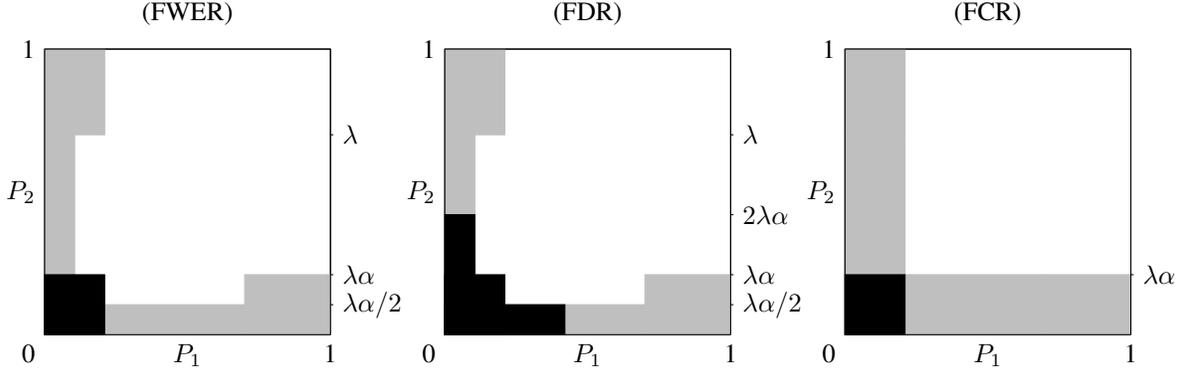

Viewed from the holistic perspective, we see that $S$ plays two distinct roles in conditional selective inference. In the first place, $S$ focuses the attention of the multiple testing procedure to hypotheses in $S$, restricting $R$ to be a subset of $S$. This is the selective property of the procedure. Secondly, by conditioning on $S$, the procedure ignores the information used to find $S$ for the final inference. This is the conditional property of the procedure. 
We see both roles of $S$ in the procedures of the toy example in Figure \ref{fig example1-1}. The procedure never rejects hypotheses outside $S$, so it is selective. We can see that the procedures is conditional, because the procedure in each $S$-defined quadrant is a valid multiple testing procedure by itself: if we would stretch any quadrant to cover the entire unit square, we would obtain a method with valid FWER, FDR or FCR control, respectively. 

The holistic perspective allows us to decouple the selective and conditional properties of conditional selective inference. We call a procedure $U \to R$ \emph{selective on $S'$} if, surely for all $\mathrm{P} \in M$, $R \subseteq S'$. We call $U \to R$ \emph{conditional on $S''$} if it controls its error rate conditionally on $S''$, i.e., if, surely, $\mathrm{E}_\mathrm{P} [e_\mathrm{P}(R, S'') \mid S'' ] \leq \alpha$. By design, a conditional selective procedure $U \to S \to R$ is selective on $S$ and conditional on $S$. However, the same procedure may be selective or conditional on sets it was not constructed around. Procedures are always selective on sets that are surely larger than $S$, and every procedure is, trivially, selective on $R$. Every procedure that is conditional on $S$ is also conditional on $U \setminus S$, since $S$ and $U \setminus S$ carry the same information. In Figure \ref{fig example1-holistic} we may verify that all three procedures are conditional and selective on, for example, $S' = \{i\colon P_i \leq (1+\lambda)/2\}$. 

In an important special case, every procedure is selective on $U$, since $R \subseteq U$ by definition. Moreover, every procedure is conditional on $U$, since the conditional error rate for $U$ is the unconditional error rate, and control of any conditional error rate implies control of the unconditional error rate. This brings us to our first main observation: For every conditional selective multiple testing procedure on $S$ there exists a conditional selective procedure on $U$, i.e.\ an unconditional, non-selective procedure, that always rejects at least as many hypotheses. 

\begin{observation} \label{thm main}
Let $U \to S \to R$ be a conditional selective inference procedure with the property that $R \subseteq S$ surely, and that $\mathrm{E}_\mathrm{P}[e_\mathrm{P}(R)\mid S] \leq \alpha$, surely, for all $\mathrm{P} \in M$. Then there exists a procedure $U \to R'$ such that $R' \supseteq R$ surely, and $\mathrm{E}_\mathrm{P}[e_\mathrm{P}(R')]  = \mathrm{E}_\mathrm{P}[e_\mathrm{P}(R')\mid U] \leq \alpha$ for all $\mathrm{P} \in M$. \end{observation}

To prove Observation \ref{thm main}, simply take $R' = R$ and observe that $\mathrm{E}_\mathrm{P}[e_\mathrm{P}(R)] = \mathrm{E}_\mathrm{P}[\mathrm{E}_\mathrm{P}\{e_\mathrm{P}(R) \mid S\}]$. We call Observation \ref{thm main} an observation rather than a theorem or proposition, because as a mathematical result it is completely trivial: if we do not restrict to $R\subseteq S$ but allow the method also to reject hypotheses in $U \setminus S$, it may achieve more rejections that way; if we do not condition on $S$, we retain more information for finding a possibly larger $R$. Observation \ref{thm main} is merely an immediate consequence of the holistic perspective we have adopted.

However, Observation 1 answers the important question how much of the information in the data to allocate to the selection step $U\to S$ and how much to the rejection step $S \to R$. According to Observation \ref{thm main}, the optimal choice is always simply to take $S=U$. Without losing power, we can allocate zero information to the selection step, and retain all of our information for the rejection step. This is an important insight.

\section{First example: the toy example} \label{sec toy}

Observation 1 says that a holistic method $U \to R'$ always exists that is at least as powerful, in the sense that $R' \supseteq R$, as a conditional selective procedure $U \to S \to R$. However, it does not show that it is always possible to achieve a true improvement, nor does it show how to find such an improvement if it exists. However, there are many cases in which substantial improvement over a conditional selective procedure is possible. 

In this section we will illustrate this with the toy example of Figure \ref{fig example1-1}, focusing on its FDR-controlling variant. The toy example will help to build an intuition for the general case. As a preview, Figure \ref{fig example1-fdr} displays the FDR-controlling conditional selective procedure (top-left), with two uniform improvements top-right and bottom-left. The bottom-left procedure is not selective on $S$, sometimes rejecting hypotheses outside $S$, but still controls FDR conditional on $S$. The top-right procedure still selective on $S$, guaranteeing $R \subseteq S$, but only has unconditional FDR control. The standard MABH procedure is given at bottom right for comparison. 

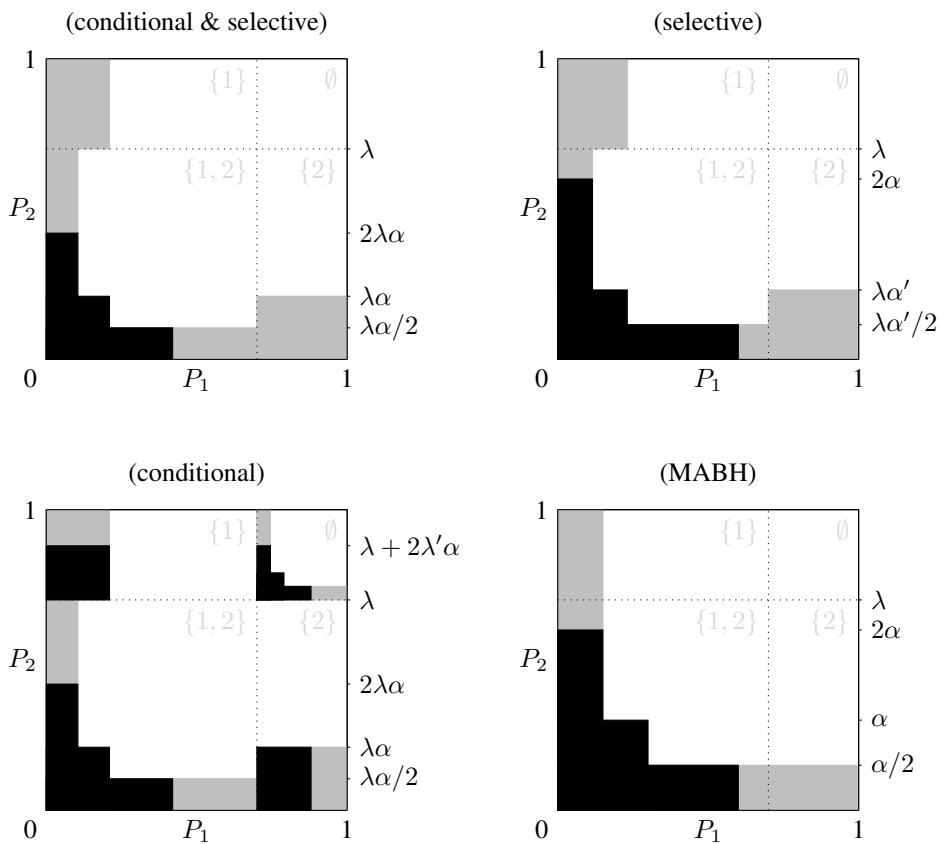
\begin{figure}[!ht]
\centering
\def\lam{.7}
\def\restlam{.3}
\def\alp{.3}
\def\factor{1.1} 

\begin{tikzpicture}[scale=4]

\begin{scope}[yshift = -1.5 cm]
\draw (0,0) rectangle (1,1);

\filldraw[gray!50] (0, \lam) rectangle (\lam*\alp, 1);
\filldraw[gray!50] (\lam, 0) rectangle (1,\lam*\alp) ;
\filldraw[gray!50] (\alp*\lam, 0) rectangle (\lam,\lam*\alp/2);
\filldraw[gray!50] (0,\alp*\lam) rectangle (\lam*\alp/2,\lam);
\filldraw[gray!50] (\lam,\lam) rectangle (1, \lam+\restlam*\alp/2);
\filldraw[gray!50] (\lam,\lam) rectangle (\lam+\restlam*\alp/2, 1);
\filldraw (0,0) rectangle (\lam*\alp,\lam*\alp);
\filldraw (0,0) rectangle (\lam*\alp/2,2*\lam*\alp);
\filldraw (0,0) rectangle (2*\lam*\alp,\lam*\alp/2);
\filldraw (\lam,\lam) rectangle (\lam +\restlam*\alp,\lam + \restlam*\alp);
\filldraw (\lam,\lam) rectangle (\lam + \restlam*\alp/2,\lam + 2*\restlam*\alp);
\filldraw (\lam,\lam) rectangle (\lam + 2*\restlam*\alp,\lam + \restlam*\alp/2);
\filldraw (\lam,0) rectangle (\lam +2*\restlam*\alp,\lam*\alp);
\filldraw (0,\lam) rectangle (\lam*\alp,\lam + 2*\restlam*\alp);
\draw (1,\lam) -- (1.01,\lam) node[right]{$\lambda$};
\draw (1,\lam*\alp) -- (1.01,\lam*\alp) node[right]{$\lambda\alpha$};
\draw (1,\lam*\alp/2) -- (1.01,\lam*\alp/2) node[right]{$\lambda\alpha/2$};
\draw (1,2*\lam*\alp) -- (1.01,2*\lam*\alp) node[right]{$2\lambda\alpha$};
\draw (1,\lam + 2*\restlam*\alp) -- (1.01,\lam + 2*\restlam*\alp) node[right]{$\lambda+2\lambda'\alpha$};

\draw (0,0) -- (1,0) node[midway, below]{$P_1$};
\draw (0,0) -- (0,1) node[midway, left]{$P_2$};
\draw (1, 0) -- (1, 1);
\draw (0,1) -- (1, 1) node[midway, above, yshift=.2cm]{(conditional)};
\draw (0,0) node[below left]{0};
\draw (0,1) node[left]{1};
\draw (1, 0) node[below]{1};

\draw[dotted] (\lam, 0) -- (\lam, 1) node[below left, gray!30]{$\{1\}$};
\draw[dotted] (0,  \lam) -- (1, \lam) node[below left, gray!30]{$\{2\}$};
\draw[dotted] (1, 1) node[below left, gray!30]{$\emptyset$};
\draw[dotted] (\lam, \lam) node[below left, gray!30]{$\{1,2\}$};
\end{scope}

\begin{scope}[xshift = 0 cm]
\draw (0,0) rectangle (1,1);

\filldraw[gray!50] (0, \lam) rectangle (\lam*\alp, 1);
\filldraw[gray!50] (\lam, 0) rectangle (1,\lam*\alp) ;
\filldraw[gray!50] (\alp*\lam, 0) rectangle (\lam,\lam*\alp/2);
\filldraw[gray!50] (0,\alp*\lam) rectangle (\lam*\alp/2,\lam);
\filldraw (0,0) rectangle (\lam*\alp,\lam*\alp);
\filldraw (0,0) rectangle (\lam*\alp/2,2*\lam*\alp);
\filldraw (0,0) rectangle (2*\lam*\alp,\lam*\alp/2);
\draw (1,\lam) -- (1.01,\lam) node[right]{$\lambda$};
\draw (1,\lam*\alp) -- (1.01,\lam*\alp) node[right]{$\lambda\alpha$};
\draw (1,\lam*\alp/2) -- (1.01,\lam*\alp/2) node[right]{$\lambda\alpha/2$};
\draw (1,2*\lam*\alp) -- (1.01,2*\lam*\alp) node[right]{$2\lambda\alpha$};

\draw (0,0) -- (1,0) node[midway, below]{$P_1$};
\draw (0,0) -- (0,1) node[midway, left]{$P_2$};
\draw (1, 0) -- (1, 1);
\draw (0,1) -- (1, 1) node[midway, above, yshift=.2cm]{(conditional \&\ selective)};
\draw (0,0) node[below left]{0};
\draw (0,1) node[left]{1};
\draw (1, 0) node[below]{1};

\draw[dotted] (\lam, 0) -- (\lam, 1) node[below left, gray!30]{$\{1\}$};
\draw[dotted] (0,  \lam) -- (1, \lam) node[below left, gray!30]{$\{2\}$};
\draw[dotted] (1, 1) node[below left, gray!30]{$\emptyset$};
\draw[dotted] (\lam, \lam) node[below left, gray!30]{$\{1,2\}$};
\end{scope}

\begin{scope}[xshift = 1.7 cm]
\draw (0,0) rectangle (1,1);

\filldraw[gray!50] (0, \lam) rectangle (\factor*\lam*\alp, 1);
\filldraw[gray!50] (\lam, 0) rectangle (1,\factor*\lam*\alp) ;
\filldraw[gray!50] (\alp*\lam, 0) rectangle (\lam,\factor*\lam*\alp/2);
\filldraw[gray!50] (0,\alp*\lam) rectangle (\factor*\lam*\alp/2,\lam);
\filldraw (0,0) rectangle (\factor*\lam*\alp,\factor*\lam*\alp);
\filldraw (0,0) rectangle (\factor*\lam*\alp/2,2*\alp);
\filldraw (0,0) rectangle (2*\alp,\factor*\lam*\alp/2);
\draw (1,\lam) -- (1.01,\lam) node[right]{$\lambda$};
\draw (1,\factor*\lam*\alp) -- (1.01,\factor*\lam*\alp) node[right]{$\lambda\alpha'$};
\draw (1,\factor*\lam*\alp/2) -- (1.01,\factor*\lam*\alp/2) node[right]{$\lambda\alpha'/2$};
\draw (1,2*\alp) -- (1.01,2*\alp) node[right]{$2\alpha$};

\draw (0,0) -- (1,0) node[midway, below]{$P_1$};
\draw (0,0) -- (0,1) node[midway, left]{$P_2$};
\draw (1, 0) -- (1, 1);
\draw (0,1) -- (1, 1) node[midway, above, yshift=.2cm]{(selective)};
\draw (0,0) node[below left]{0};
\draw (0,1) node[left]{1};
\draw (1, 0) node[below]{1};

\draw[dotted] (\lam, 0) -- (\lam, 1) node[below left, gray!30]{$\{1\}$};
\draw[dotted] (0,  \lam) -- (1, \lam) node[below left, gray!30]{$\{2\}$};
\draw[dotted] (1, 1) node[below left, gray!30]{$\emptyset$};
\draw[dotted] (\lam, \lam) node[below left, gray!30]{$\{1,2\}$};
\end{scope}

\begin{scope}[xshift = 1.7 cm, yshift=-1.5cm]
\draw (0,0) rectangle (1,1);

\filldraw[gray!50] (0, 0) rectangle (\alp/2, 1);
\filldraw[gray!50] (0, 0) rectangle (1,\alp/2) ;
\filldraw (0,0) rectangle (\alp,\alp);
\filldraw (0,0) rectangle (\alp/2,2*\alp);
\filldraw (0,0) rectangle (2*\alp,\alp/2);
\draw (1,\lam) -- (1.01,\lam) node[right]{$\lambda$};
\draw (1,\alp) -- (1.01,\alp) node[right]{$\alpha$};
\draw (1,\alp/2) -- (1.01,\alp/2) node[right]{$\alpha/2$};
\draw (1,2*\alp) -- (1.01,2*\alp) node[right]{$2\alpha$};

\draw (0,0) -- (1,0) node[midway, below]{$P_1$};
\draw (0,0) -- (0,1) node[midway, left]{$P_2$};
\draw (1, 0) -- (1, 1);
\draw (0,1) -- (1, 1) node[midway, above, yshift=.2cm]{(MABH)};
\draw (0,0) node[below left]{0};
\draw (0,1) node[left]{1};
\draw (1, 0) node[below]{1};

\draw[dotted] (\lam, 0) -- (\lam, 1) node[below left, gray!30]{$\{1\}$};
\draw[dotted] (0,  \lam) -- (1, \lam) node[below left, gray!30]{$\{2\}$};
\draw[dotted] (1, 1) node[below left, gray!30]{$\emptyset$};
\draw[dotted] (\lam, \lam) node[below left, gray!30]{$\{1,2\}$};
\end{scope}
\end{tikzpicture}
\caption{The conditional selective procedure of the toy example, controlling FDR, with its conditional and selective improvements. The MABH procedure is given as reference. Grey indicates areas in which one hypothesis is rejected; black indicates areas in which both are rejected. Here, $\lambda' = 1-\lambda$, and $\alpha' = \alpha/(2\lambda - \lambda^2)$.}
\label{fig example1-fdr}
\end{figure}

How did we arrive at these improvements? For the conditional improvement (bottom left), we keep aiming for control of FDR conditional on $S$, but we allow the procedure to reject hypotheses in $U \setminus S$. To do this, we also calculate selection-adjusted $p$-values $P_{i|S}$ for $i \notin S$. We obtain
\begin{equation} \label{eq conditional P}
P_{i|S} = \left\{ \begin{array}{ll} P_i/\lambda & \textrm{if $i \in S$} \\
(P_i-\lambda)/(1-\lambda) & \textrm{if $i \notin S$.} \end{array} \right. 
\end{equation}
While the selection-adjusted $p$-values are larger than the original ones for $i \in S$, the reverse is true when $i \notin S$. Next, we extend the procedure by continuing to test hypotheses in $U\setminus S$ after all hypotheses in $S$ are rejected. If $S = \{1,2\}$ the procedure is not changed. If $S= \{1\}$ and $H_1$ was rejected, we may continue to test $H_2$, rejecting when $P_{2|S=\{1\}}\leq 2\alpha$, and analogous for $S= \{2\}$. This fixed-sequence procedure (conditional on $S=\{1\}$) is easily seen to be valid for FDR control, and is related to fixed-sequence FDR-controlling procedures by \cite{farcomeni2013} and \cite{lynch2017control}. If $S = \emptyset$, rather than rejecting nothing, we may use a MABH procedure on $P_{1|S=\emptyset}$ and $P_{2|S=\emptyset}$. 

The resulting procedure, quite a strange one, is given at bottom-left in Figure \ref{fig example1-fdr}. It consists of four miniature multiple testing procedures, applied to conditional $p$-values, and valid conditional on $S$ for the four realizations of $S$. For $S=\{1,2\}$ and $S=\emptyset$ we have a conditional MABH; for $S=\{1\}$ and $S=\{2\}$ was have a fixed-sequence FDR-controlling procedure, prioritizing the hypothesis in $S$. The resulting procedure clearly uniformly improves the procedure of Figure \ref{fig example1-1}. It does this by also considering hypotheses outside $S$ for rejection. However, the improved procedure retains the property that it controls FDR conditional on $S$, since each of the miniature procedures is valid for FDR control.

A different type of improvement may be achieved if we are willing to give up on conditional FDR control. This is shown in the top-right of Figure \ref{fig example1-fdr}. The improvement comes in two parts. First, we remark that the original procedure does not exhaust the $\alpha$-level under the global null hypothesis: if $H_1 \cap H_2$ is true, FDR is controlled at level $(2\lambda-\lambda^2)\alpha$. We can therefore gain power by starting the procedure at level $\alpha' = \alpha/(2\lambda-\lambda^2)$ instead of at $\alpha$. Secondly, after the original procedure has rejected $H_1$, it rejects $H_2$ if $P_{2|S=\{1,2\}} \leq 2\alpha$, i.e., when $P_{2} \leq 2\lambda\alpha$. If we are not doing conditional control, however, there is no need to use the conditional $p$-value, and we may alternatively reject $H_2$, after we have rejected $H_1$, simply if $P_{2} \leq 2\alpha$. The procedure resulting from these two improvements is given at the top-right of Figure \ref{fig example1-fdr}. The procedure's FDR control is not conditional on $S$ anymore, but it remains selective on $S$, assuming $\lambda \geq 2\alpha$. The validity of this new procedure may not be immediately obvious; we prove it in in the following lemma.

\begin{lemma}
Suppose $P_1, P_2$ are independent and standard uniform under $H_1, H_2$, respectively. Without loss of generality, assume that $P_1 \leq P_2$. Let $0\leq \lambda \leq 1$ and $\alpha' = \alpha/(2\lambda - \lambda^2)$. Define a procedure that that rejects $H_1$ when $P_1 \leq \lambda\alpha'/2$, or when $P_1 \leq \lambda\alpha'$ and $P_2 \leq \lambda\alpha'$, or when $P_1 \leq \lambda\alpha'$ and $P_2 > \lambda$, and that rejects $H_2$ when $H_1$ is rejected and $P_2 \leq 2 \alpha$. This procedure controls FDR at level $\alpha$.
\end{lemma}

\begin{proof}
We prove FDR control separately the cases that 2, 1, or 0 null hypotheses are true. We remark that the procedure is visualized on the right-hand side of Figure \ref{fig example1-fdr}.

Suppose $H_1$ and $H_2$ are true. Then the false discovery proportion is 1 whenever at least one rejection occurs. Since $P_1$ and $P_2$ are independent and standard uniform, this probability can be checked to be
\[
2(1-\lambda)\lambda\alpha' + 2\lambda \lambda \alpha' /2 - (\lambda \alpha'/2)^2 + (\lambda \alpha'/2)^2 =
\alpha.
\]
so FDR is $\alpha$.

Suppose $H_1$ is true, but $H_2$ is not. Then the false dicovery proportion is 1/2 in the black area of Figure \ref{fig example1-fdr}, 0 in the upper left grey area, and 1 in the lower right grey area. If $P_2 \leq 2\alpha$, then $\mathrm{E}_\mathrm{P}(\mathrm{FDP}|P_2\leq 2\alpha) \leq 1/2 \mathrm{P}(P_1\leq 2\alpha|P_2\leq 2\alpha) = (1/2)\cdot 2\alpha = \alpha$. If $P_2 > 2\alpha$, then $\mathrm{E}(\mathrm{FDP}|P_2> 2\alpha) \leq \mathrm{P}(P_1\leq \lambda\alpha'|P_2 > 2\alpha) = \lambda\alpha' = \lambda\alpha/(2\lambda-\lambda^2) \leq \alpha$. It follows that FDR is at most $\alpha$. The case that $H_2$ is true, but $H_1$ is not is analogous. 

If $H_1$ and $H_2$ are both false then FDP is always 0 and there is nothing to prove.
\end{proof}

We have constructed two improvements of the conditional selective procedure we started with. One of the procedures retains the property of the original procedure that it controls FDR conditional on $S$, the second retains the property that it only rejects hypotheses in $S$. The holistic perspective, however, does not care about $S$ or about properties relating to $S$. It sees these two new methods simply as uniform improvements of the original that never reject fewer hypotheses and sometimes more. One of these, the bottom-left one, is arguably somewhat weird and difficult to motivate from a holistic perspective (compare \cite{berger1989uniformly}'s tests improving the likelihood ratio test and \citeauthor{perlman1999emperor}'s (\citeyear{perlman1999emperor}) discussion); the top-right one seems more reasonable.

As a fourth procedure, bottom-right in Figure \ref{fig example1-fdr}, we have given the regular MABH procedure, that does not attempt to be conditional or selective on $S$. This might be the procedure we would have chosen if we would have adopted a holistic perspective from the beginning. In this particular case, MABH actually happens to be selective on $S$ (as long as $\lambda \geq 2\alpha$). Comparing the conditional procedure (bottom-left) to MABH, we see a massive shift of power away from $S=\{1,2\}$ towards $S=\{1\}$, $S=\{2\}$, and $S=\emptyset$. Comparing the selective procedure (top-right) to MABH, we see that, while both procedures are selective, the original MABH still focuses relatively more power on $S=\{1,2\}$; the top-right procedure still has a relatively large focus on small sets $S$. This focus actually chimes with the motivation of the procedure we started from: \cite{zhao2019} and \cite{ellis2020} advocated their method for an application context in which null $p$-values tend to be near 1, so that $S=\{1\}$ and $S=\{2\}$ are relatively likely.  

The comparison with MABH also serves to illustrate that uniformly improving a method $U \to S \to R$ by $U \to U \to R'$, with the requirement that $R' \supseteq R$, is not usually a question of simply adjusting the tuning parameter $\lambda$ in such a way that $S$ becomes $U$. The MABH procedure (bottom-right), resulting from the choice $\lambda=1$ in the conditional selective method (top-left), will be a more powerful method in many situations, but is not a uniform improvement of the original method unless $\lambda \leq 1/2$. Generally, finding a true uniform improvement, in the sense that $R' \supseteq R$ surely for all $\mathrm{P} \in M$, involves much more work than merely adjusting a tuning parameter.

Comparing the conditional selective procedure and its two improvements, we see that the conditional selective procedure is exactly the intersection of its conditional and its selective improvement: it rejects either of $H_1, H_2$ if and only if both the selective and the conditional improvements do. Compared to the conditional selective procedure, the selective improvement may have additional rejections if $S=\{1,2\}$, while the conditional improvement cannot. On the other hand, the conditional improvement may have more rejections if $S=\emptyset$, while the selective procedure remains powerless there. If $S=\{1\}$ or $S=\{2\}$, both procedures may have additional rejections compared to the conditional selective procedure. However, the selective procedure has more chance of rejecting the hypothesis in $S$, while the conditional procedure may additionally reject a hypothesis outside $S$. The two improvements are, in this sense, disjoint.

The two improvements in Figure \ref{fig example1-fdr} are easy to generalize to the case of more than two null hypotheses. They illustrate an important general principle about selection and conditioning in multiple testing. This principle says that selection and conditioning each pull a procedure in opposite directions. Conditioning forces a procedure to distribute its power evenly over the outcome space, since the procedure must have proper error control on all realizations of $S$, conditional on $S$. Selection, on the other hand, focuses the power of procedures away from hypotheses in $U\setminus S$, since it restricts rejections to $S$. A procedure that is both selective and conditional must therefore necessarily focus power both away from $S$ and away from $U \setminus S$. Since there is nowhere for the power to go, it vanishes. The conditional selective procedure at top left, being the intersection of a conditional and a selective procedure, is therefore sub-optimal as either. It is definitely sub-optimal from the holistic perspective.

\section{(In)admissibility conditions}

Having looked in detail at a small example, we will now come back to the general case. We will give some sufficient conditions under which uniform improvements exist.

We call a conditional selective inference procedure $U \to S \to R$ \emph{inadmissible} if $U \to R'$ exists that uniformly improves upon $U \to S \to R$ in the sense that $R \subseteq R'$, surely for all $\mathrm{P} \in M$, and $\mathrm{P}(R \subset R') > 0$ for at least one $\mathrm{P} \in M$, while still controlling the error rate, i.e., $\mathrm{E}_\mathrm{P}[e_\mathrm{P}(R')] \leq \alpha$. We will be a bit more precise and call $U \to S \to R$ \emph{inadmissable as a selective method on $S$} if the uniform improvement still satisfies $R' \subseteq S$, surely. Similarly, we call $U \to S \to R$ \emph{inadmissable as a conditional method on $S$} if the uniform improvement still controls its error rate conditional on $S$. Remember, however, that from the holistic perspective we do not care too much about $S$ or about these sub-classes of inadmissibility. 

Our definition of a uniform improvement is very strict \citep[as in][]{goeman2021}, requiring that $R' \subseteq R$ for every outcome $\omega \in \Omega$. A uniform improvement, therefore, can never fail to reject a hypothesis that the method it improves does reject. This requirement makes admissibility a very low bar to achieve. For example, a FWER-controlling method that rejects all hypotheses in $U$ with probability $\alpha$, independently of the data, and rejects nothing with probability $1-\alpha$, is admissible according to our definition. Since admissibility is so easy to achieve, inadmissibility is particularly bad news.

We will give several sufficient conditions for inadmissibility of conditional selective methods. Propositions \ref{thm inadmissible TP}, \ref{thm inadmissible empty} and \ref{thm extra sup} apply to any error rate. Proposition \ref{thm inadmissible FWER} is only for FWER control.

\begin{proposition} \label{thm inadmissible TP}
If $\delta>0$ is known such that $\mathrm{P}[S \cap T_\mathrm{P} = \emptyset] \geq \delta$ for all $\mathrm{P} \in M$, then $U \to S \to R$ is inadmissible as a selective procedure on $S$, unless $R=S$ surely for all $\mathrm{P} \in M$.
\end{proposition}

\begin{proof}
Let $q = (1-\alpha - \delta) / (1-\alpha)(1- \delta)$ if $\delta < 1 - \alpha$ and $q=0$ otherwise. Let $R' = R$ with probability $q$, and $R' = S$ otherwise. Then, since $e_\mathrm{P}(S) = e_{\mathrm{P}}(R) = 0$ if $S \cap T_{\mathrm{P}} = \emptyset$, we have $\mathrm{E}_{\mathrm{P}}[e_\mathrm{P}(R)] \leq \alpha \mathrm{P}[S \cap T_\mathrm{P} \neq \emptyset] \leq \alpha(1-\delta)$, and $\mathrm{E}_{\mathrm{P}}[e_\mathrm{P}(R)] \leq \mathrm{P}[S \cap T_\mathrm{P} \neq \emptyset] \leq (1-\delta)$. Therefore, 
\[
\mathrm{E}_\mathrm{P}[e_\mathrm{P}(R')] = q\mathrm{E}_\mathrm{P}[e_\mathrm{P}(R)] + (1-q)\mathrm{E}_\mathrm{P}[e_\mathrm{P}(S)] \leq q(1-\delta)\alpha + (1-q)(1-\delta) = \alpha.
\] 
It follows that $R'$ controls the unconditional error rate.

Noting that $q <1$, we have $\mathrm{P}(R \subset R') = (1-q) \mathrm{P}(R \subset S) > 0$ for at least one $\mathrm{P} \in M$ unless $R=S$ surely for all $\mathrm{P} \in M$, so $R'$ uniformly improves over $R$.

Finally, we have trivially that $R \subseteq S$ surely for all $\mathrm{P} \in M$.

It follows that $U \to S \to R$ is inadmissible as a selective procedure on $S$.
\end{proof}

In words, Proposition \ref{thm inadmissible TP} says that any conditional selective procedure is inadmissible if, with positive probability, the selection step results in a set $S$ without true hypotheses \citep[for examples, see][]{ellis2020, al2020adaptive, heller2023simultaneous}. In this case, it is impossible to make false discoveries, and the $\alpha$ for such $S$ can be better spent elsewhere. The condition of the proposition implies that $S$ has FWER control at level $\delta$, but allows $\delta > \alpha$. The proposition does not apply when $R=S$ surely, but we come back to that case in Observation \ref{thm R=S} in Section \ref{sec FCR testing}.

\begin{proposition} \label{thm inadmissible empty}
If $\mathrm{P}(S=\emptyset)>0$ for some $\mathrm{P} \in M$, then $U \to S \to R$ is inadmissible as a conditional procedure on $S$. It is inadmissible as a selective procedure on any $S'$ for which $S' \supseteq S$ surely for all $\mathrm{P} \in M$, and $S' \neq \emptyset$ surely for all $\mathrm{P} \in M$. 
\end{proposition}

\begin{proof}
Choose any $S'$ that fulfils the assumptions, noting that $S' = U$ always fits. Let $R' = R$ if $S \neq \emptyset$, and if $S = \emptyset$, let $R' = S'$ with ancillary probability $\alpha$ and $R'=\emptyset$ otherwise. Then by assumption there exists $\mathrm{P} \in M$ such that $\mathrm{P}(R \subset R') =  \alpha \mathrm{P}(S = \emptyset) > 0$. By conditional error control, we have $\mathrm{E}_\mathrm{P}[e_\mathrm{P}(R') \mid S] = \mathrm{E}_\mathrm{P}[e_\mathrm{P}(R) \mid S] 
\leq \alpha$ if $S \neq \emptyset$. If $S = \emptyset$, we have $\mathrm{E}_\mathrm{P}[e_\mathrm{P}(R') \mid S] \leq \alpha$ by construction, since $e_\mathrm{P}(\emptyset) = 0$. This proves inadmissibility as a conditional method on $S$. Noting that conditional control on $S$ implies unconditional control, and that $R' \subset S'$ surely, we have inadmissibility as a selective method on $S'$.
\end{proof}

Proposition \ref{thm inadmissible empty} says that a conditional selective procedure may be improved if it sometimes selects $S = \emptyset$. There is a subtle but important difference with Proposition \ref{thm inadmissible TP}: if $\mathrm{P}(S=\emptyset)>0$ for all $\mathrm{P} \in M$, then we would fulfil the conditions for Proposition \ref{thm inadmissible TP}, but Proposition \ref{thm inadmissible empty} only requires that this happens for at least one $\mathrm{P} \in M$. Intuitively, if $S=\emptyset$ sometimes, we can make no errors in that case, and we can spend the $\alpha$ allocated to that case elsewhere.

\begin{proposition} \label{thm extra sup}
If $\alpha'$ is known such that
\begin{equation} \label{eq extra sup}
\alpha' = \sup_{\mathrm{P} \in M} \mathrm{E}_\mathrm{P}[e_\mathrm{P}(R)] < \sup_{\mathrm{P} \in M} \mathrm{E}_\mathrm{P}\Big[\sup_{\mathrm{P} \in M} \mathrm{E}_\mathrm{P}[e_\mathrm{P}(R)\mid S)]\Big],
\end{equation}
and $\mathrm{P}(R=S)<1$ for at least one $\mathrm{P} \in M$, then $U \to S \to R$ is inadmissible as a selective method. 
\end{proposition}

\begin{proof}
Note that $\sup_{\mathrm{P} \in M} \mathrm{E}_\mathrm{P}[e_\mathrm{P}(R)\mid S)] \leq \alpha$ for all $S$, since $U \to S \to R$ controls its error rate conditionally. Therefore, the right-hand side of (\ref{eq extra sup}) is at most $\alpha$. Let $\alpha' = \sup_{\mathrm{P} \in M} \mathrm{E}_\mathrm{P}[e_\mathrm{P}(R)]$, so $\alpha' < \alpha$. Let $R' = R$ with probability $1-q$, and $R=S$ with probability $q$, where $q = (\alpha - \alpha')/(1-\alpha')$. Then we have
\[
\sup_{\mathrm{P} \in M} \mathrm{E}_\mathrm{P}[e_\mathrm{P}(R')] = (1-\alpha+\alpha') \sup_{\mathrm{P} \in M} \mathrm{E}_\mathrm{P}[e_\mathrm{P}(R)] + (\alpha-\alpha') \sup_{\mathrm{P} \in M} \mathrm{E}_\mathrm{P}[e_\mathrm{P}(S)] \leq (1-q) \alpha' + q = \alpha,
\]
so $U \to R'$ controls the error rate (unconditionally). Since $R' \subseteq S$ surely, $U \to R'$ is selective on $S$. By the condition of the proposition, there is a $\mathrm{P} \in M$ such that $\mathrm{P}(R' \supset R) = q \mathrm{P}(R \neq S) > 0$. It follows that $U \to R'$ uniformly improves upon $U \to S \to R$ as a selective method, so $U \to S \to R$ is inadmissible.
\end{proof}

To understand Proposition \ref{thm extra sup}, note that the left-hand side of (\ref{eq extra sup}) is equal to 
\[
\sup_{\mathrm{P} \in M} \mathrm{E}_\mathrm{P}\big[ \mathrm{E}_\mathrm{P}[e_\mathrm{P}(R)\mid S)]\big],
\]
so that (\ref{eq extra sup}) holds with $\leq$ by definition. Unconditional control bounds the left-hand side of (\ref{eq extra sup}) by $\alpha$, while conditional control implies that the right-hand side of (\ref{eq extra sup}) is bounded by $\alpha$. Any gap between the two can be exploited by an unconditional test to gain power. Such a gap may arise if the `worst case' $\mathrm{P}$, for which the conditional $\alpha$-level is exhausted, depends on $S$. We give an example in  Appendix \ref{sec extra sup}. 

\begin{proposition} \label{thm inadmissible FWER}
If $U \to S \to R$ controls FWER conditional on $S$, and there exists $\mathrm{P} \in M$ such that $\mathrm{P}(R = S \mid S) > 0$ for some $S \subset U$, then $U \to S \to R$ is inadmissible as a conditional procedure on $S$.
\end{proposition}

\begin{proof}
With probability $1-\alpha$, let $R' = R$, and with probability $\alpha$, let $R'= U$ if $R = S$, and $R' = R$ otherwise.
We will prove that $R'$ controls FWER conditional on $S$ for every $\mathrm{P} \in M$. We have either $T_{\mathrm{P}} \cap S = \emptyset$ or $T_{\mathrm{P}} \cap S \neq \emptyset$. In the former case,  $\mathrm{P}(R' \cap T_{\mathrm{P}} \neq \emptyset \mid S) \leq \mathrm{P}(R'\supseteq R) \leq \alpha$, since it is not possible to make a Type I error with $R \subseteq S$. In the latter case, $\mathrm{P}(R' \cap T_{\mathrm{P}} \neq \emptyset \mid S) = \mathrm{P}(R \cap T_{\mathrm{P}} \neq \emptyset \mid S) \leq \alpha$, since $R$ controls FWER conditional on $S$. It follows that $R'$ controls FWER conditional on $S$ for every $\mathrm{P} \in M$. 
According to the assumption, there exists  $\mathrm{P} \in M$ such that $\mathrm{P}(R \subset R') = \mathrm{P}(R = S \subset U) > 0$. It follows that $R$ is inadmissible as a conditional procedure on $S$.
\end{proof}

Proposition \ref{thm inadmissible FWER} exploits the Sequential Rejection Principle \citep{goeman2010}, which says that if we reject all hypotheses under consideration, we may recycle the $\alpha$ and continue testing with a new batch. For a conditional selective procedure, this means that if we have exhausted all hypotheses in $S$, we may continue testing hypotheses in $U \setminus S$.

In the toy example, we see that the conditions of Propositions \ref{thm inadmissible TP}, \ref{thm inadmissible empty} and \ref{thm inadmissible FWER}
are all fulfilled, provided that $\lambda < 1$. The probability that we select only false null hypotheses is $(1-\lambda)^2, 1-\lambda$, or 1 respectively in the situation that 2, 1 or 0 hypotheses are true, so the condition of Proposition 1 is fulfilled with $\delta = (1-\lambda)^2$. Under $\mathrm{P} \in H_1 \cap H_2$ we have $\mathrm{P}(S = \emptyset) = (1-\lambda)^2 > 0$, so also the condition of Proposition \ref{thm inadmissible empty} is fulfilled. Finally, if FWER was controlled, take $S = \emptyset$; then all hypotheses in $S$ are rejected with positive probability for every $\mathrm{P} \in M$, conditional on $S=\emptyset$. It may seem from this checking of the conditions that the crucial characteristic that makes the procedure in the toy example inadmissible is the fact that it selects $S=\emptyset$ with positive probability. However, this is not the only driving factor. For example, perhaps the most important improvement of the top-left over the top-right procedure in Figure \ref{fig example1-fdr} is the increase of the critical value from $2\lambda\alpha$ to $2\alpha$ for rejecting the second hypothesis after rejecting the first. This change is not tied to the selection of $S=\emptyset$ in any way. The propositions of this section are sufficient conditions for inadmissibility, but they are by no means necessary. We will see examples of improvements of procedures that never select $S=\emptyset$ in Sections \ref{sec winner} and \ref{sec data splitting general}.

The propositions in this section should be seen as examples of classes of procedures that might be improved by letting go of selection and conditioning. The emphasis was on uniform improvements. Often, procedures may be constructed that do not necessarily uniformly improve upon the original, but are substantially more powerful for relevant alternatives. An example is the standard MABH in the toy example, which, although not a uniform improvement over the original, has much larger rejection regions for both $H_1$ and $H_2$.

\section{Second example: conditioning on the winner} \label{sec winner}

The toy example that we considered thus far may have seemed to hinge much on the property that it selected $S=\emptyset$ with positive probability. Here, we look at a situation in which $\mathrm{P}(S=\emptyset)=0$ for all $\mathrm{P} \in M$.

The hypotheses that attract most attention in publications are generally those with smallest $p$-values. It is of interest, therefore, to consider selection rules based on ranks. Selective inference for such selections, ``inference on winners'', has been considered by \citet{zhong2008, reid2017, fuentes2018, zrnic2020, andrews2022, zrnic2022locally}. We consider the simplest set-up here, where we select only a single ``winner''. In this set-up, we consider the question whether the winner is truly non-null.

Let $P_1, \ldots, P_n$ be independent $p$-values, standard uniform under their respective null hypotheses $H_1, \ldots, H_n$, so that $U = \{1,\ldots, n\}$. We consider the selection rule that selects the single hypothesis for which the $p$-value is smallest, with ties broken arbitrarily, so that $|S|=1$ always. 

If we want to condition on the selection event $S=\{i\}$, we cannot simply reject for small values of $P_i$, adjusting the critical value for the selection event as we did in the toy example of Figure \ref{fig example1-1}. To see why this would be problematic, consider a set-up with $n=2$ in which $H_1$ is null, but $H_2$ is not. Then 
\begin{equation} \label{eq conditional impossible}
\mathrm{P}(P_1 \leq t \mid S=\{1\}) =
\frac{\mathrm{P}(P_1 \leq t, P_1 \leq P_2)}{\mathrm{P}(P_1 \leq P_2)} =
\frac{\mathrm{P}(P_1 \leq P_2 \wedge t)}{\mathrm{P}(P_1 \leq P_2)} =
\frac{\mathrm{E}_\mathrm{P}(P_2 \wedge t)}{\mathrm{E}_\mathrm{P}(P_2)}.
\end{equation}
Since $P_2$ is under the alternative, its distribution is arbitrary, so it could be uniform on $[0,t]$. In that case, (\ref{eq conditional impossible}) evaluates to 1. Therefore, for every $t>0$, there exists a $\mathrm{P} \in M$ such that $\mathrm{P}(P_1 \leq t \mid S=\{1\}) = 1$. Therefore, it is impossible to bound (\ref{eq conditional impossible}), in supremum over $\mathrm{P} \in M$, by $\alpha$. Consequently, it impossible to construct a conditional selective procedure that rejects for small values of $P_i$.

A way out of this conundrum was offered by \cite{reid2017}, who proposed to use as an alternative test statistic $P_{i|S=\{i\}} = P_i/\min_{j \neq i} P_j$. Conditional on $S=\{i\}$, we have that $P_i/\min_{j \neq i} P_j$ is standard uniform for all $\mathrm{P} \in H_i$, as Lemma \ref{lem p ratio i} states. Based on this lemma we can construct a conditional selective inference procedure. It rejects $H_i$, $i \in S$, when $P_i/\min_{j \neq i} P_j \leq \alpha$. We call this Procedure A.

\begin{lemma} \label{lem p ratio i}
If $n\geq 2$, and $\mathrm{P} \in H_i$, then $P_i/\min_{j \neq i}P_j \sim \mathcal{U}(0,1)$ given $S=\{i\}$.
\end{lemma}

\begin{proof}
Choose any $\mathrm{P} \in H_i$. We have
\[
\mathrm{P}\Big(\frac{P_{i}}{\min_{j \neq i}P_j} \leq t \,\Big|\, S=\{i\}, \min_{j \neq i}P_j=q\Big) 
= \mathrm{P}\big(P_{i} \leq qt \mid P_i \leq q, \min_{j \neq i} P_{j} = q\big) = t,
\]
where we use that $\min_{j \neq i} P_{j}$ and $P_i$ are independent. Taking expectations conditional on $S=\{i\}$ on both sides, the result follows.
\end{proof}

What error rate does this conditional procedure on $S$ control? On a family $S$ of only one hypothesis, unadjusted testing, FCR, FWER and FDR control are all identical; Procedure A, therefore, controls all these error rates simultaneously. To construct potential improvements of the method, we must, therefore, decide which error rate to retain control of. We choose FDR for this example.

As in Section \ref{sec toy} we will construct three alternative procedures. The first, Procedure B, retains validity conditional on $S$, but possibly rejects hypotheses outside $S$. The second, Procedure C, will have unconditional FWER control, but still only rejects hypotheses within $S$. The third procedure, Procedure D, will be fully unconditional and defined on $U$.

To construct procedure $B$, we must extend the notion of conditional $p$-values for $H_j$, $j \notin S$. We need the following lemma.

\begin{lemma} \label{lem p ratio j}
If $n\geq 2$, and $\mathrm{P} \in H_j$, $j \neq i$, then $(P_j-P_i)/(1-P_i) \sim \mathcal{U}(0,1)$, independent of $(P_k)_{k \neq j}$, given $S=\{i\}$. 
\end{lemma}

\begin{proof}
Choose any $\mathrm{P} \in H_j$. We have
\[
\mathrm{P}(\frac{P_{j}-P_{i}}{1-P_{i}} \leq t \mid S=\{i\}, (P_k)_{k \neq j} = (q_k)_{k \neq j}) 
= \mathrm{P}(\frac{P_{j}-q_{i}}{1-q_{i}} \leq t \mid P_j > q_i, (P_k)_{k \neq j} = (q_k)_{k \neq j}) 
= t,
\]
where we use that $(P_k)_{k \neq j}$ and $P_j$ are independent. Taking expectations on both sides, we have the required unconditional uniformity. Since the conditional probability does not depend on $(P_k)_{k \neq j}$, it follows that $(P_j-P_i)/(1-P_i)$ is independent of these $p$-values. 
\end{proof}

We will use $P_{i\mid S} = (P_j-P_i)/(1-P_i)$ for $j \neq i$. As in Section \ref{sec toy}, we see that adjustment for non-selection results in $p$-values that are smaller than their unadjusted counterparts, rather than larger. Procedure B will be a two-step method based on these selection-adjusted $p$-values. Let $i$ be such that $S=\{i\}$. Then, first, the procedure tests $H_i$, rejecting if $P_i/\min_{j \neq i} P_j \leq \alpha$. If it fails to reject $H_i$, the procedure stops. Otherwise it continues with a BH-procedure at level $\alpha' = n\alpha/(n-1)$ on the $n-1$ hypotheses $H_j$, $j \neq i$, using $(P_j-P_i)/(1-P_i)$ as $p$-values. This procedure clearly uniformly improves upon Procedure A if $n>1$. The validity of this procedure is proved by Lemma \ref{lem FDR} below.

\begin{lemma} \label{lem FDR}
Procedure B controls FDR given $S=\{i\}$.
\end{lemma}

\begin{proof}
Let $R$ denote the rejected set of Procedure B. We condition on $S=\{i\}$. Choose any $\mathrm{P} \in M$.
We either have $\mathrm{P} \in H_i$, or $\mathrm{P} \notin H_i$. If $\mathrm{P} \in H_i$, then by Lemma \ref{lem p ratio i} we have that $\mathrm{P}(P_i/\min_{k \neq i} P_k \leq \alpha) \leq \alpha$, so $R = \emptyset$ with probability $1-\alpha$, so FWER is controlled given $S=\{i\}$, so FDR is controlled given $S=\{i\}$. If $\mathrm{P} \notin H_i$, then let $R'$ be the rejected set of the second step of the procedure. By Lemma \ref{lem p ratio j}, this step is applied on independent and uniform $p$-values, given $S=\{i\}$. By \cite{benjamini1995}, therefore, \[
\mathrm{E}_\mathrm{P}\Big(\frac{|R' \cap T_\mathrm{P}|}{|R'|\vee 1}\mid S=\{i\}\Big) \leq \alpha'.
\]
Since $R$ is either the empty set or $R' \cup \{i\}$, we have, using that $i \notin T_\mathrm{P}$ and $|R'|\leq n-1$,
\[
\frac{|R \cap T_\mathrm{P}|}{|R| \vee 1} = \frac{|R' \cap T_\mathrm{P}|}{|R'| + 1} = \frac{|R'|}{|R'|+1}\frac{|R' \cap T_\mathrm{P}|}{|R'| \vee 1} \leq \frac{n-1}{n}\frac{|R' \cap T_\mathrm{P}|}{|R'| \vee 1}.
\]
It follows that $$\mathrm{E}_\mathrm{P}\Big(\frac{|R \cap T_\mathrm{P}|}{|R|\vee 1} \mid S=\{i\}\Big) \leq \frac{n-1}{n} \mathrm{E}_\mathrm{P}\Big(\frac{|R' \cap T_\mathrm{P}|}{|R'|\vee 1}\mid S=\{i\}\Big) \leq \frac{n-1}{n} \alpha' = \alpha,$$
so Procedure B also controls FDR given $S=\{i\}$ when $\mathrm{P} \notin H_i$.
\end{proof}

For Procedure C, we ignore the conditioning on $S=\{i\}$, but still restrict rejection to $S$ only. This means that we can simply reject $H_i$ for small $P_i$. By independence of the $p$-values, we may reject $H_i$ when $P_i \leq 1-(1-\alpha)^{1/n}$. This is Procedure C. For Procedure D, the fully unconditional procedure, we simply choose the familiar BH-procedure. 

While Procedure B uniformly improves upon procedure A, the unconditional Procedures C and D do not. To see this, consider the situation that $P_2, \ldots, P_n$ are always equal to 1 (which they could be under the alternative, or if null $p$-values are allowed to be stochastically larger than uniform). In that case, Procedures A and B reject $H_1$ if $P_1 \leq \alpha$, while Procedures C and D need $P_1 \leq 1-(1-\alpha)^{1/n}$ and $\alpha/n$, respectively.

\begin{figure}[!ht]
    \centering
    \includegraphics[width=\textwidth]{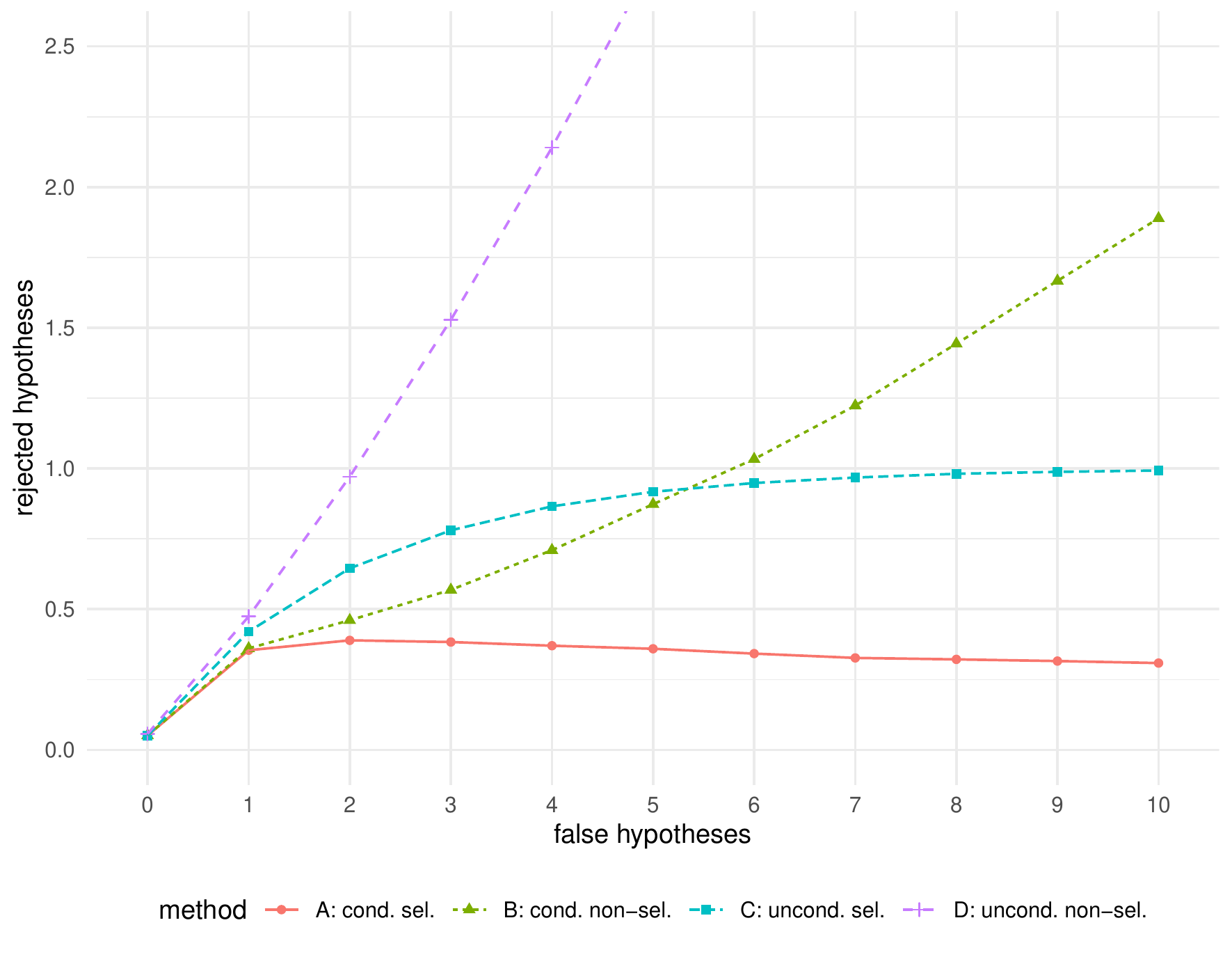}
    \caption{Expected number of rejections for four procedures defined in Section \ref{sec winner}. Based on $n=100$ hypotheses, $\alpha=0.05$, and $10^4$ simulations.} 
    \label{fig sim winner}
\end{figure}

We compared the four procedures in a simple simulation. Out of 100 hypotheses, from 0 to 10 were considered to be under the alternative, getting a $p$-value based on a one-sided normal test with a mean shift of 3; the remaining $p$-values were standard uniform. Figure \ref{fig sim winner} reports the expected number of rejected hypotheses for each of the methods A, B C and D. We see that the original conditional Procedure A is very much directed toward sparse alternatives, even losing power as the density of the signal increases. In contrast, all other methods gain power with increasing signal. The unconditional Procedure C, which like Procedure A only ever rejects the winner, rejects it with larger probability than Procedure A for all scenarios. The fully unconditional BH method, although not a uniform improvement, is the clear overall winner, rejecting most hypotheses on average even in the sparse scenarios.

\section{Data splitting and carving} \label{sec data splitting general}

Data splitting is perhaps the archetypal conditional selective inference method. It splits the data into two parts, using the first part for selecting $S$, and the second part for inference. Standard data splitting splits the data by subjects. Data carving is a more advanced version of data splitting \citep{fithian2014, panigrahi2018, schultheiss2021} that uses alternative ways of splitting the information in the data into independent parts, and use the data more efficiently that way. We show that data splitting and carving are inadmissible in general, at least for FWER control.

A special feature of data splitting is that the selection step that results in $S$ is completely unconstrained, as long as the selection remains independent of the second part of the data. This implies that the universe $U$ from which $S$ was chosen is in principle infinite. The inadmissibility conditions of Section \ref{sec holistic} still apply, however. We have a simple corollary to Proposition \ref{thm inadmissible FWER}, due to the infinite nature of $U$. The inefficiency of data splitting has been noted by other authors. \cite{jacobovic2022} established inadmissibility of \citeauthor{moran1973}'s \citeyearpar{moran1973} data-split test, and \cite{fithian2014} have shown that data splitting yields inadmissible selective
tests in exponential family models.

\begin{proposition} \label{thm data splitting}
Data splitting is inadmissible as a selective method for FWER control if $U$ is infinite and $S$ is almost surely finite.
\end{proposition}

\begin{proof}
By Proposition \ref{thm inadmissible FWER}, since $S \neq U$ almost surely, it is sufficient to show that $\mathrm{P}(R=S) > 0$ for some $\mathrm{P} \in M$. 
Choose any $V \subseteq U$. We will show that $\mathrm{P}(R=S | S=V) > 0$ for some $\mathrm{P} \in M$. Conditional on $S=V$, the set $R$ is the result of a procedure with conditional FWER control on $V$. We write the procedure as a sequential rejection procedure \citep{goeman2010}. Let $\mathcal{N}$ be the next function in that formulation. Suppose that the procedure is admissible, and that $\mathrm{P}(R=V | S=V)=0$ for all $\mathrm{P} \in M$. We will derive a contradiction. Let $W \subset V$ be the largest set such that $\mathrm{P}(R=W|S=V)>0$ for at least one $\mathrm{P} \in M$. Then the procedure is equivalent to a procedure that has $\mathcal{N}(W) = \emptyset$ almost surely. A uniform improvement is, therefore, a procedure that has $\mathcal{N}'(W) = \{i\}$ with probability $\alpha$, where $i$ is the smallest element of $S \setminus W$. This is a uniform improvement, since the probability that the new procedure rejects more is $\alpha\mathrm{P}(R=W|S=V) > 0$. To check that the new procedure retains FWER control given $S=V$, we need to check the monotonicity and single step conditions of \cite[Theorem 1]{goeman2010}, both of which are trivial.
It follows that the procedure we started with is inadmissible, and we have the contradiction we need.   
\end{proof}

Proposition \ref{thm data splitting} says that a data splitting procedure is inadmissible because the analyst always runs the risk of selecting too few hypotheses for $S$. If all hypotheses in $S$ are rejected, the classic data splitting procedure must stop, and loses out on some rejections it could have made. A uniform improvement would be a procedure that selects not just $S$, but an infinite sequence of pairwise disjoint continuations $S_1, S_2, \ldots$. This procedure would always continue testing the next selected set after the previous one has been completely rejected. All of $S_1, S_2, \ldots$ must still be chosen using the first part of the data only. Control is, therefore, still conditional on the first part of the data. 

Proposition \ref{thm data splitting} speaks about FWER control only. We conjecture that the same result holds for FDR, since FDR by its nature is more lenient than FWER for making further rejections (in $S_2, S_3,\ldots$) if has already made many rejections (all of $S_1$). We do not have a general proof for this, but, as an example, consider FDR-controlling methods of the type discussed by \citet{li2017accumulation}. These estimate FDR along an incremental sequence of potential rejection sets, rejecting the largest set for which the FDR estimate is less than $\alpha$. Such procedures would gain power if the sequence is continued beyond $S$ into $S_1, S_2, \ldots$.

With data splitting, the split of the data in two parts is arbitrary by nature, and the question how much of the data to use for the selection and inference steps arises naturally. Some authors have proposed repeated splitting \citep{meinshausen2009, diciccio2020}. Such methods are unconditional: while inference in each random split is conditional on the $S$ from that split, control in the final analysis unconditional. Multiple data splitting can, therefore, also be seen as an unconditional improvement of a conditional method.

\section{Third example: data splitting}

In Section \ref{sec data splitting general}, we showed that data splitting is inadmissible as a conditional method for FWER control. If we are prepared to move away from conditional control, we can often improve methods further, although not always uniformly. We investigate a specific simple case in more detail.

Let $U = \{1, \ldots, n\}$ be finite, and suppose the analysis on the two parts of the data results in pairs of independent $p$-values $P_{1,i}, P_{2,i}$, for $H_i$, for $i=1,\ldots, n$. A natural choice for $S$ is $S = \{i\colon P_{1,i} \leq \lambda\}$ for some fixed $0 \leq \lambda \leq 1$. With this choice, a conditional Bonferroni procedure would reject
\begin{equation} \label{eq data split cond}
R  = \{i \in S\colon P_{2,i} \leq \alpha/|S|\}.
\end{equation}
We can rewrite this as $R  = \{i \in U\colon Q_i \leq \lambda \alpha/|S|\}$, with $Q_i = \lambda P_{2,i}$, if $P_{1,i} \leq \lambda$, and $Q_i = 1$, otherwise. Here, $Q_i$ is a valid unconditional $p$-value, since
$
\mathrm{P}(Q_i \leq t) = \mathrm{P}(P_{1,i} \leq \lambda)\mathrm{P}(\lambda P_{2,i} \leq t) = \lambda \min(t/\lambda, 1) \leq t.
$
We could also have constructed an unconditional procedure on $U$ based on the same $Q_i$. This would reject
\begin{equation} \label{eq data split uncond}
R'  = \{i \in U\colon Q_i \leq \alpha/n\}.
\end{equation}

Comparing the conditional and unconditional procedures (\ref{eq data split cond}) and (\ref{eq data split uncond}), we see that $R' \subseteq R$ whenever $|S| \leq \lambda n $, and $R' \supseteq R$ otherwise. The conditional procedure, seemingly, only has a chance to reject more than the unconditional if $|S|$ is smaller than its expectation under the complete null hypothesis with uniform $p$-values. The more signal in the data, the larger we would expect $S$ to be, and the smaller the conditional $R$ becomes relative to the unconditional $R'$. The conditional procedure only has a chance to be better only if null $p$-values are stochastically larger than uniform. This argument generalizes immediately beyond Bonferroni to other symmetric monotone procedures. E.g., the unconditional procedure of \citet{benjamini1995} on $Q_i$, $i \in U$ dominates its conditional equivalent on $Q_i$, $i \in S$ if $|S| > \lambda n$.

In the example just discussed, with $S = \{i\colon P_{1,i} \leq \lambda\}$, if $\lambda$ was fixed a priori and $P_{1,1}, \ldots, P_{1,n}$ are independent, then we are not using all the information remaining after selecting $S$. Rather than splitting the data into $P_{1,1}, \ldots, P_{1,n}$ used for finding $S$ and $P_{2,1}, \ldots, P_{2,n}$ used for testing, the data can be split into $1_{\{P_{1,1} \leq \lambda\}}, \ldots, 1_{\{P_{1,n} \leq \lambda\}}$  used for finding $S$ and $P_{1,1\mid S}, \ldots, P_{1,n\mid S}$ and $P_{2,1}, \ldots, P_{2,n}$ used for testing. Such an alternative splits are known as data carving. They tune the amount of information that is allocated to the selection and testing steps more efficiently. However, from the perspective of unconditional procedures, this still seems a rather convoluted way of combining the information from $P_{1,i}$ and $P_{2,i}$. A natural and more powerful choice would be, e.g., a Fisher combination, equivalent to rejecting for low values of $P_{1,i}\times P_{2,i}$, or, even more naturally, a single $p$-value calculated form a direct analysis of the combined data. Such analyses also obviate the need for choosing $\lambda$.

\section{Selective confidence intervals and the False Coverage Rate} \label{sec FCR}

So far we have focused mostly on rejection of hypotheses based on $p$-values. However, a large part of the selective inference literature focuses on selection-adjusted confidence intervals, controlling the (conditional) FCR. In this section we will apply the holistic perspective to selective inference based on confidence intervals.

A confidence interval is a random subset $C \subseteq M$ of the model space $M$. A confidence interval is said to have $(1-\alpha)$-coverage if, for all $\mathrm{P} \in M$, 
\[
\mathrm{P}(\mathrm{P} \in C) \geq 1-\alpha.
\]
We define confidence intervals always as a subset of the full parameter space. We can do this without loss of generality. For example, if our parameter space for $\theta = (\theta_1, \theta_2)$ is $\mathbb{R}^2$, we can write the confidence interval $[a,b]$ for $\theta_1$ as the ``interval'' $C = [a,b] \times \mathbb{R}$ for $\theta$. This greatly simplifies notation. We keep using the word interval, though $C$ can be any region. 

In the selective inference context, we have $S \subseteq U$ be a random set of confidence intervals of interest, where $U$, as before, is the universe from which we are selecting. The collection of confidence intervals depends on $S$, and we write $C_{i\mid S}$, $i \in S$. The confidence intervals should have 
conditional $(1-\alpha)$-coverage if, for all $\mathrm{P} \in M$, and for $i \in S$,
\begin{equation} \label{eq sel coverage}
\mathrm{P}(\mathrm{P} \in C_{i\mid S}\mid S) \geq 1-\alpha.
\end{equation}

If we report more than one confidence interval we must account for multiplicity. We can demand that the confidence intervals are (conditionally) \emph{simultaneous over the selected}, i.e., surely for all $\mathrm{P} \in M$,
\begin{equation} \label{eq simsel}
\mathrm{P}\Big(\mathrm{P} \in \bigcap_{i\in S} C_{i\mid S} \,\big\vert\, S\Big) \geq 1-\alpha,
\end{equation}
where the unconditional variant drops the conditioning on $S$. Similarly, we can control FCR. The unconditional variant demands that,  for all $\mathrm{P} \in M$, 
\begin{equation} \label{eq CI FCR uncond}
\mathrm{E}_P \bigg[ \frac{|\{i \in S\colon \mathrm{P} \in C_{i\mid S}\}|} {|S| \vee 1} \bigg] \geq 1-\alpha.
\end{equation}
Conditional on $S$, this simplifies to the demand that, surely for all $\mathrm{P} \in M$, 
\begin{equation} \label{eq CI FCR}
\frac1{|S| \vee 1} \sum_{i\in S} \mathrm{P}\Big(\mathrm{P} \in C_{i\mid S} \,\big\vert\, S\Big) \geq 1-\alpha.
\end{equation}
It is one of the attractive properties of selection-adjusted confidence intervals that they control FCR without further adjustment, since (\ref{eq sel coverage}) implies (\ref{eq CI FCR}); see also \citet{weinstein2013}, Theorem 2; \citet{lee2016}, Lemma 2.1; \citet{fithian2014}, Proposition 11.

For confidence intervals we have the following analogue of Observation \ref{thm main}. 

\begin{observation} \label{thm simsel}
If $C_i$, $i \in S$, control (\ref{eq simsel}) or (\ref{eq CI FCR}) conditionally on $S$, then there exist $C'_i$, $i \in U$, such that $C'_i \subseteq C_i$ for $i \in S$ surely, that control (\ref{eq simsel}) or  (\ref{eq CI FCR}), respectively, with $S = U$.   
\end{observation}

This observation is, again, trivial. We simply take $C'_i = C_i$ if $i \in S$ and $C'_i = M$ otherwise. Like Observation \ref{thm main}, Observation \ref{thm simsel} answers the question what the optimal choice of $S$ is, if we are interested in confidence intervals that are as narrow as possible. The answer is that $S=U$ is the optimal choice.

Like Observation \ref{thm main}, Observation \ref{thm simsel} does not say whether taking $S=U$ can actually help to shorten the confidence intervals. However, it is easy to find examples in which this is possible, certainly for FCR control. Take, for example, the original FCR-controlling method of \cite{benjamini2005}, which constructs marginal confidence intervals of level $1-|S|\alpha/|U|$. For this method, taking $S=U$ clearly results in the narrowest confidence intervals. This observation holds generally for FCR control: as confidence intervals tends to become narrower as $S$ becomes larger, there is every incentive for the analyst to choose $S$ as large as possible, since they will obtain both more and narrower confidence intervals. In the extreme case that $S=U$, FCR control reduces to average marginal coverage, an even weaker criterion than marginal coverage, which is achieved by uncorrected confidence intervals. 

Specifically for the property of simultaneous over the selected, we have the following additional observation.

\begin{observation}
If $C_i$, $i \in S$, are unconditionally simultaneous over the selected $S$, then for every $S' \subseteq U$, there exists $C'_i$, $i \in S'$, which are unconditionally simultaneous over the selected $S'$, such that $C'_i \subseteq C_i$ surely for all $i \in S \cap S'$. \end{observation}

To see that this observation is true, simply take $C'_i = C_i$ for $i \in S \cap S'$, and $C_i' = M$ for $i \in S' \setminus S$.

The observation says that any unconditional method that is simultaneous on the selected for some $S \subseteq U$, is also simultaneous on the selected on any other $S' \subseteq U$. This suggests, at least for unconditional methods, that simultaneous on the selected is not a different concept from just simultaneous over $U$, i.e., simultaneous.

\section{Fourth example: post-selection inference for the lasso}

One of the major application areas of conditional selective inference is post-selection inference on the parameters of a lasso model. A major breakthrough here has been the polyhedral lemma \citep{lee2016}, which allows calculation of $p$-values and confidence intervals for regression coefficients, conditional on their selection by a lasso algorithm. The toy example of Section \ref{sec toy} is in fact a special case of the approach of \cite{lee2016}, and we will not discuss that again. In this section we consider a variant due to \citet{liu2018} of lasso-based selective inference, in which additional interesting issues arise. 

The set-up is as follows. We assume the usual linear model setting, in which we have a fixed $n \times m$ design matrix $X$, and assume that $Y = X\beta + \epsilon$, where $\beta$ (an $m$-vector) is unknown, and $\epsilon \sim \mathcal{N}(0,\sigma^2 I_n)$, where $\sigma^2$ is assumed known. In this model we fit a lasso regression with a fixed penalty parameter $\lambda$. Let $\tilde\beta_i$, $i=1,\ldots,m$, be the resulting coefficient estimates. We define the selected set as $S = \{i\colon \tilde\beta_i \neq 0\}$. 

\citet{liu2018} define selection-adjusted confidence intervals by not conditioning on the full selected set $S$, but only on the selection of the confidence interval of interest. They require that, for all $\mathrm{P} \in M$, and for $i \in S$,
\begin{equation} \label{eq liu}
\mathrm{P}(\mathrm{P} \in C_{i\mid i\in S}\mid i \in S) \geq 1-\alpha.
\end{equation}
Condition (\ref{eq liu}), while implied by (\ref{eq sel coverage}), is substantially weaker, because it conditions on less information. In a part of their paper \cite{fithian2014} considered conditioning on $i \in S$, rather than on the full $S$ for testing, recognizing that less conditioning leads to more information for inference. \citet{liu2018} adopted this viewpoint for confidence intervals, arguing that by conditioning on this minimal event, more variation remains in the data for determining the precise value of $\beta_i$. The methodology of \citet{jewell2019} and \cite{Neufeld2022} shares the `general recipe' of \cite{liu2018}, stating that the ultimate goal is to fulfill equation (\ref{eq liu}) rather than (\ref{eq sel coverage}) when it comes to selective inference.

Indeed, the conceptual difference between the two properties (\ref{eq liu}) and (\ref{eq sel coverage}) is huge, but there is a steep price to pay for conditioning only on $i \in S$. Complications arise in subsequent error rate control because the coverage of each $C_{i\mid i\in S}$ is conditional on a different event for every $i\in S$. Because of this, the property, mentioned in Section \ref{sec FCR}, that selection-adjusted coverage (\ref{eq sel coverage}) implies FCR control (\ref{eq CI FCR}), is lost: (\ref{eq liu}) does not imply (\ref{eq CI FCR}) or even (\ref{eq CI FCR uncond}). Without a common conditioning event, there is no hope for combining the confidence intervals into any combined conditional error rate. For example, making $|S|$ confidence intervals, each conditional on $j \in S$, at level $1-\alpha/|S|$ does not guarantee simultaneous coverage, even unconditionally; we need confidence intervals at level $1-\alpha/m$ for that. In Appendix \ref{sec_extra_fourthexample} we give a numerical example showing lack of conditional and unconditional FCR control of the confidence intervals of \cite{liu2018} at confidence level $1-\alpha$, and lack of conditional and unconditional simultaneous control at confidence level $1-\alpha/|S|$. Lack of FCR control of the method of \cite{liu2018} was also observed by \citet[Table 1]{panigrahi2022approximate}, but without explanation.

By Observation \ref{thm simsel}, there is no reason to be selective and report confidence intervals for $i \in S$ only. Indeed, the premise of restricting attention to the selection of $S$ is often that variables not in $S$ are not important for the outcome.  Confidence intervals or $p$-values for non-selected variables are an important instrument to check this. It is straightforward to extend the theory of \cite{liu2018} to calculate $C_{i\mid i\notin S}$, $i \notin S$, for the non-selected regression coefficients, and we give the mathematical details in Appendix \ref{sec_extra_fourthexample}. Figure \ref{fig intervals} display 90\%-confidence intervals for all eight variables of the famous Prostate data set \citep{stamey1989prostate} as a function of $\lambda$, with intervals for selected coefficients in black and for non-selected ones in grey. We see a similar paradoxical effect as in the toy example: conditional intervals of selected variables tend to move towards 0, while confidence intervals for non-selected variables tend to move away from 0 (see also Figure \ref{fig pvalues} in Appendix B). Both are equal to the unconditional intervals for very large or small $\lambda$, when the probability of selection is close to 0 or 1, but tend to become longer close to the critical threshold for selection. \cite{kivaranovic2020, kivaranovic2021}  provide conditions under which intervals obtained from the polyhedral lemma are either bounded or unbounded. The intervals constructed through the method of \citet{liu2018} have bounded lengths when they are conditional on selection, whereas the intervals are potentially unbounded when they are conditional on non-selection.

\begin{figure}[!ht]
    \centering
    \includegraphics[width=0.8\textwidth]{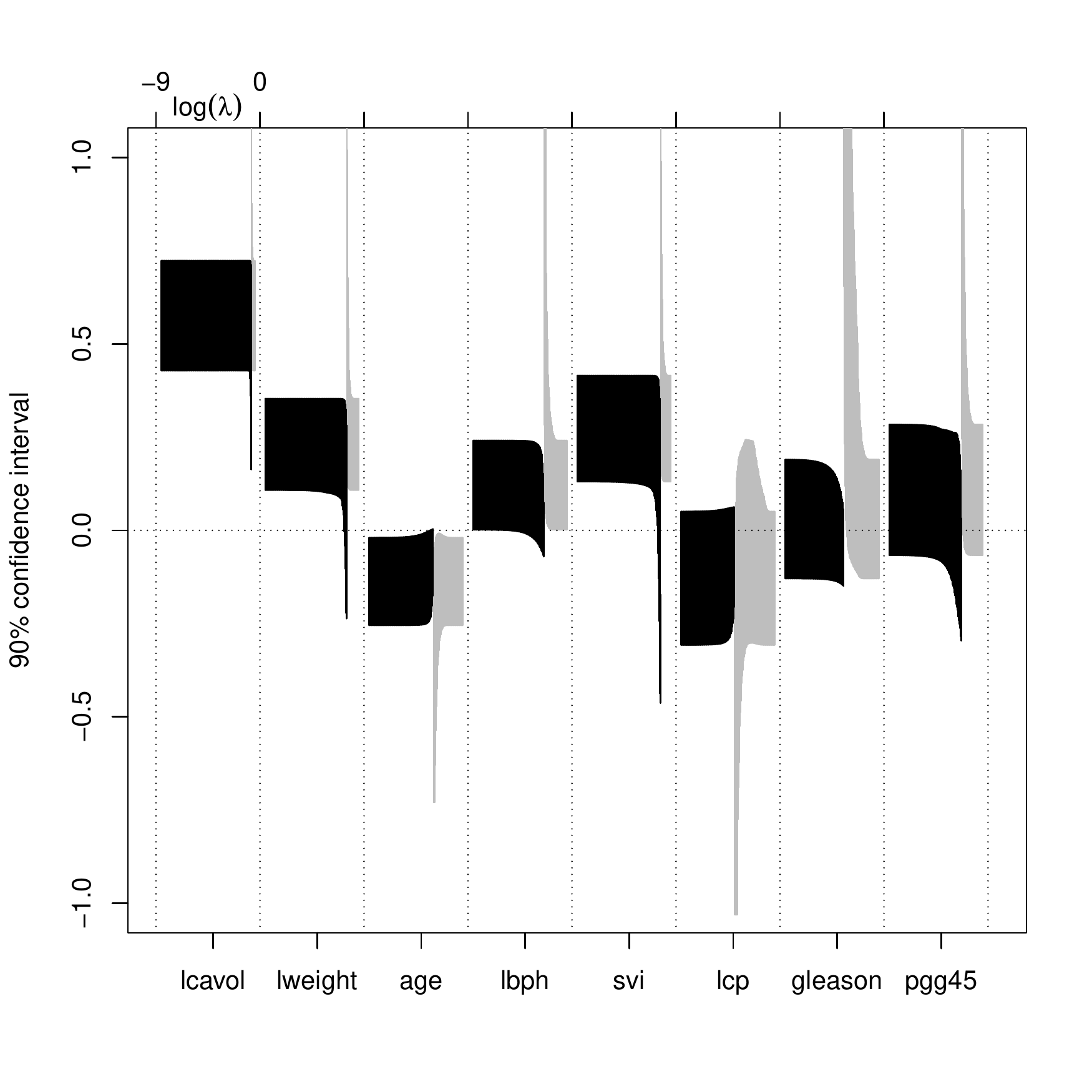}
    \caption{Selective conditional confidence intervals using the method of \cite{liu2018} applied to the variables of the Prostate data, as a function of $\lambda$. Black intervals are conditional on selection by the lasso, grey ones are conditional on non-selection.} 
    \label{fig intervals}
\end{figure}

The intervals $C_i$, defined as $C_{i\mid i\in S}$, if $i \in S$, and $C_{i\mid i\notin S}$, if $i \notin S$, are unconditional intervals and, due to the absence of a common conditioning event, have no conditional interpretation as a collection. We may present them all as uncorrected intervals, but if we aim to present only a selection $V \subseteq \{1,\ldots, m\}$ from these intervals we must correct for this using methods to correct unconditional intervals. We may use level $1-\alpha/m$ to obtain simultaneous coverage over the selected intervals, or we may use the method of \cite{benjamini2005} and use level $1-|V|\alpha/m$ to control FCR. This applies if $V=S$ or for any other $V$. There is no way in which the conditioning of the intervals on $i\in S$ helped for this correction step; in fact, it merely discarded valuable information, lengthening the intervals and moving them towards zero. Arguably, the superior method is simply to start from regular unconditional intervals. This does not provide a uniform improvement of the method of \cite{liu2018}, but it avoids the paradoxes associated with conditioning, and tends to produce more attractive intervals.

\section{FCR for hypothesis testing} \label{sec FCR testing}

Confidence intervals can be used to test hypotheses, and the properties of confidence intervals imply error control guarantees on the hypotheses. Here, we look briefly into the error rate (\ref{eq FCR}) implied by FCR control (\ref{eq CI FCR}), which is used by some authors \citep{fithian2014}. Assume that we have a collection $H_i$, $i \in S$, of hypotheses, one for every confidence interval. 

If confidence intervals $C_{i\mid S}$, $i \in S$, have conditional FCR control, then $R = \{i\colon H_i \cap C_{i\mid S} = \emptyset\}$ controls the error rate (\ref{eq CI FCR}). Observation \ref{thm main} does not directly apply, since the error rate depends not just on $R$ in $S$. However, that observation immediately generalizes.

\setcounter{observation}{0}
\begin{observation}[\textbf{continued}] 
Observation 1 also holds for error rates $e_\mathrm{P}(R, S)$ that depend on $S$, if $S \subseteq S'$ implies that $e_\mathrm{P}(R, S) \geq e_\mathrm{P}(R, S')$.
\end{observation}
\setcounter{observation}{3}

The extra condition holds for the FCR rate (\ref{eq FCR}). The condition implies that replacing $S$ by $U$ makes the error rate more lenient, so for controlling the error rate it helps to take $S=U$, and the result is still trivial. FCR is a paradoxical error rate from the holistic perspective, since it is decreasing in $|S|$ for the same $R$. This gives an immediate incentive for an analyst to choose $S$ as large as possible. 

FCR is sometimes motivated \citep{zhao2020} by the property that FCR control reduces to FDR control when $S=R$. For this property to hold, we must have that $S=R$ as random variables; it is not sufficient that the realised values are identical. About conditional control of FCR (or other error rates) when $S=R$ as random variables we have the following observation. We call a testing problem trivial on $\mathcal{S}$ if $e_\mathrm{P}(V) \leq \alpha$ for all $\mathrm{P} \in M$ and all $V \in \mathcal{S}$, i.e., if the error rate is already bounded by $\alpha$ everywhere. 

\begin{observation} \label{thm R=S}
Suppose a conditional selective method $U \to S \to R$ has $R=S$ surely. Then the testing problem is trivial on $\mathcal{S}$.
\end{observation}

To see that this observation is true, remark that conditional control requires that $\mathrm{E}_\mathrm{P}[e_\mathrm{P}(R)\mid S] \leq \alpha$ for all $\mathrm{P} \in M$ and all $S \in \mathcal{S}$. If $R=S$ surely, the inequality reduces to $e_\mathrm{P}(S) \leq \alpha$.

It follows from Observation \ref{thm R=S} that only unconditional FCR controlling methods can be used as a means to construct FDR controlling methods; conditional FCR control has no relationship to FDR control.

\section{Discussion}

The literature on selective inference methods based on conditioning often takes the selected set of hypotheses $S$ as given, and presents the analyst's task solely as providing confidence intervals or $p$-values that are valid despite the random nature of $S$. In this paper, we see this as only the middle step of a bigger procedure, that first selects $S$ from a universe $U$, corrects for this selection, and finally uses the resulting $p$-values or confidence intervals to control an error rate of choice, resulting in a final rejected set $R$. This holistic perspective is perhaps the most important contribution of this paper. All the results in this paper are tied to this perspective.

If $S$ is simply a step in a procedure that starts with a universe $U$ and ends with a rejected set $R$, the question arises naturally what is the optimal amount of information to invest in choosing $S$. The simple answer is: none. For both primary roles of $S$, i.e.\ automatically accepting hypotheses not in $S$, and discarding all information used to select $S$, the optimal choice is to choose $S$ as large as possible.

Selection-adjusted $p$-values of confidence intervals are sometimes presented as the end result of a conditional selective inference procedure, suggesting that selection-adjustment is sufficient to address the multiplicity problem. However, the error rate (\ref{eq FCR}) thus controlled is equivalent to the per-comparison error rate (i.e\ unadjusted testing) on $S$. It does not correct for the multiplicity of $S$ itself. The larger $S$, therefore, the more and the lower the selection-adjusted $p$-values will be. From the holistic perspective, there is every incentive for the analyst to choose $S$ as large as possible, eventually reaching unadjusted testing when $S=U$. In our view, it is appropriate to present selection-adjusted $p$-values or confidence intervals without further multiple testing adjustment only if the choice of $S$ not under the control of the analyst, and only if unadjusted methods would have been appropriate if $S$ would have been non-random and given a priori.

We have given several examples of uniform improvements of conditional methods by unconditional ones, as well as general conditions under which such improvements are possible. Some of these improvements are useful and substantial; others are small or may appear artificial. We do not have a general recipe for such improvements, and we emphasize that improvements are generally not unique. In several case studies we have constructed improved procedures that were either still selective, i.e., focusing power on a small and promising set $S$ of hypotheses, or still conditional, i.e., valid conditional on the information used to find this same $S$. Invariably, we found that good selective procedures were not conditional, and good conditional procedures were not selective. Apparently, prioritizing hypotheses in $S$ and conditioning on this prioritization are conflicting goals. A multiple testing procedure that focuses its power on a promising set $S$ should exploit the information that $S$ is a promising set; a conditional procedure discards the same information by conditioning on it.

Choosing $S=U$, as we advocate, essentially means reverting to unconditional, as opposed to more stringent conditional error rates. In our view this is good enough: common unconditional error rates such as familywise error are seldom criticized for being too lenient. Some authors \citep[e.g.,][]{kuffner2018} have argued that it is better to control conditional error rates because they avoid unwarranted use of ancillary information. We find this difficult to accept as a general argument, since in most procedures $S$ is not ancillary in the usual sense, but based on a bonafide summary of the available evidence in part of the data. 

Finally we remark that allowing looks at the data prior to making inferential decisions is not exclusively the domain of conditional methods. In fact, simultaneous methods allow users to postpone some inferential decisions until after seeing all of the data \citep{goeman2011}, something conditional methods could never allow.

\appendix

\section*{Supplemental Information}

Appendices below will be Supplemental Information. Proofs will also be moved from the main text to the Supplemental information, but we left them in the main text for now for the benefit of the reviewing process.

All code to run the examples in this paper is available at \url{https://aldosolari.github.io/selectingconditioning/}.

\section{Fifth example: a uniform improvement related to Proposition \ref{thm extra sup}} \label{sec extra sup}

We give an example of a uniform improvement of a procedure that relates to Proposition \ref{thm extra sup}.

Consider a simple hypothesis testing problem with two null hypotheses about the same parameter $\mu$. Suppose that for some fixed $\delta \geq 0$ are interested in the following two hypotheses:
\[
H_1: \mu \geq -\delta; \qquad H_2: \mu \leq \delta.
\]

Suppose $X_1$ and $X_2$ are independent $\mathcal{N}(\mu, 1)$. In the spirit of data splitting, we will use $X_1$ to decide which of the two hypotheses we will test using $X_2$. Consider the following conditional procedure. Let $S= \{1\}$ if $X_1 < 0$ and $S = \{2\}$ if $X_1 \geq 0$. This seems sensible, since we expect $X_2$ to have the same sign as $X_1$ with high probability if at least one of the null hypotheses is false. Therefore, $S$ pre-selects the null hypothesis we are most likely to reject. Next we choose how to test the hypothesis in $S$. If $S = \{1\}$, we reject $H_1$ if $X_2 < -\delta -z_\alpha$, where $z_\alpha$ is chosen such that $\Phi(-z_\alpha) = \alpha$, and $\Phi$ is the standard normal distribution function; If $S = \{2\}$, we reject $H_2$ when $X_2 > \delta + z_\alpha$. It is easy to check that conditional on $S$, the probability of falsely rejecting the hypothesis in $S$ is bounded by $\alpha$, and that this probability is exactly $\alpha$ in the situation that $\mu=\delta$ if $S=\{2\}$ and if $\mu=-\delta$ if $S=\{1\}$. As a conditional selective procedure, this procedure can not be uniformly improved. Since we have $|S|=1$, as in Section \ref{sec winner}, the procedure controls FWER as well as all less stringent error rates.

We can, however, improve the procedure uniformly as an unconditional procedure. The condition of Proposition \ref{thm extra sup} is fulfilled, since for this procedure the `worst case' $\mathrm{P}$, i.e.\ the distribution for which the $\alpha$-level of the test is exhausted, depends on $S$.
Let us aim to retain FWER control, and write down the closed testing procedure that is implied by the procedure we have just constructed. Write $H_{12} = H_1 \cap H_2: -\delta \leq \mu \leq \delta$. This closed testing procedure rejects $H_{12}$ if $s(X_1)X_2 > \delta+z_\alpha$, where $s(X_1)$ is the sign of $X_1$, taken as 1 if $X_1=0$; it rejects $H_1$ if $-X_2 > \delta+z_\alpha$ and $X_1 < 0$, and $H_2$ if $X_2 > \delta+z_\alpha$ and $X_1 \geq 0$. We can check that this procedure rejects $H_1$ or $H_2$ exactly when the conditional procedure does.

Next, we check whether this procedure exhausts its $\alpha$-level. The probability of rejecting $H_{12}$ is
\begin{eqnarray*}
\mathrm{P} (s(X_1)X_2 > \delta+z_\alpha) &=& \mathrm{P}_\mu(X_2 > \delta + z_\alpha)\mathrm{P}_\mu(X_1\geq 0) + \mathrm{P}_\mu(X_2 < -\delta - z_\alpha)\mathrm{P}_\mu(X_1<0) \\
&=& \Phi(\mu - \delta - z_\alpha)\Phi(\mu) + \Phi(-\mu - \delta - z_\alpha)\Phi(-\mu). 
\end{eqnarray*}
Within $-\delta \leq \mu \leq \delta$, this is maximized when $\mu = \delta$ or $\mu = -\delta$, when we have a rejection probability of 
\[
\mathrm{P} (s(X_1)X_2 > \delta+z_\alpha) = \alpha \Phi(\delta) + \Phi(-2\delta - z_\alpha) (1-\Phi(\delta)).
\]
This probability is equal to $\alpha$ only if $\delta = 0$, and strictly smaller than $\alpha$ otherwise. Reasoning similarly, we can calculate the probability of rejecting $H_1$ or $H_2$ as bounded by a smaller factor $\alpha \Phi(\delta)$. Since none of the local tests exhaust the $\alpha$-level, we can uniformly improve the closed testing procedure by performing all tests at an increased nominal level $\alpha'$ such that
\[
\alpha' \Phi(\delta) + \Phi(-2\delta - z_{\alpha'}) (1-\Phi(\delta)) = \alpha,
\]
instead of $\alpha$. Although the difference between $\alpha'$ and $\alpha$ vanishes as $\delta \to 0$ or $\delta \to \infty$, it can be substantial in between. For example, with $\alpha = 0.05$ and $\delta =0.5$, we find $\alpha' = 0.0654$. We can increase the $\alpha$-level of the local tests of $H_1$ and $H_2$ even further to $\alpha/\Phi(\delta)$, but doing so would not improve the closed testing procedure as a whole.

\section{Supplementary material to the Fourth Example} \label{sec_extra_fourthexample}

\subsection{Mathematical details of \cite{liu2018}}

Let $\eta_i = X( X^\mathsf{T} X)^{-1}e_i$, where $e_i$ is a vector with all components equal to 0, except the $i$th, which is 1, and write $y$ as the sum of two orthogonal vectors $u_i + v_i$ with $u_i=P_{\eta_i}y$, $v_i=(I_n - P_{\eta_i})y$ and $P_{\eta_i} = \eta_i \eta_i^\mathsf{T} /\| \eta_i \|^2$. 

Proposition 1 of \cite{liu2018} shows that the distribution of $\hat{\beta}_i = \eta_i^\mathsf{T} y$ conditional to $i \in S$ and $v_i$ is the $N(\beta_i, \sigma^2 \|\eta_i\|^2)$ truncated to $(-\infty, a_i] \cup [b_i,\infty)$, where 
$$a_i = \|\eta_i\|^2 (X_i^\mathsf{T} r_i - \lambda), \quad b_i = \|\eta_i\|^2 (X_i^\mathsf{T} r_i + \lambda)$$ with $r_i = X_{-i} \tilde{\beta}_{-i} - v_i$, $X_i$ is the $i$th column of $X$, $X_{-i}$ is the submatrix of $X$ after removing the $i$th column, and $\tilde{\beta}_{-i}$ is the lasso fit of $v_i$ to $X_{-i}$ with penalty $\lambda$.
Likewise, the distribution of $\hat{\beta}_i = \eta_i^\mathsf{T} y$ conditional to $i \notin S$ and $v_i$ is the $N(\beta_i, \sigma^2 \|\eta_i\|^2)$ truncated to $(a_i,b_i)$.

\subsection{Numerical example}

In the following simulation, it is demonstrated numerically that there is a lack of FCR control at the $\alpha$ level and a lack of simultaneous control at the confidence level of $1-\alpha/|S|$.
We have considered a setting that is similar to the one used in Appendix B of \cite{liu2018}, with $n=m=2$, $X_1 = (\sqrt{(1+\rho)/2},-\sqrt{(1-\rho)/2})^\mathsf{T}$ and $X_2 = (\sqrt{(1+\rho)/2},\sqrt{(1-\rho)/2})^\mathsf{T}$. Each column of $X$ has unit norm and $X_1^\mathsf{T}X_2 = \rho$. We set $\rho=0.95$ and we simulated $10^5$ realizations of $Y\sim N(\mu, \sigma^2 I)$ by choosing $\sigma^2=1$ and $\beta = (5,5)^\mathsf{T}$, which gives $\mu = X\beta = (10,0)^\mathsf{T}$. The penalty parameter for the lasso was set to $\lambda = 0.2$ and the confidence level $(1-\alpha)$ of \citeauthor{liu2018}'s  confidence intervals was set to $90\%$. 

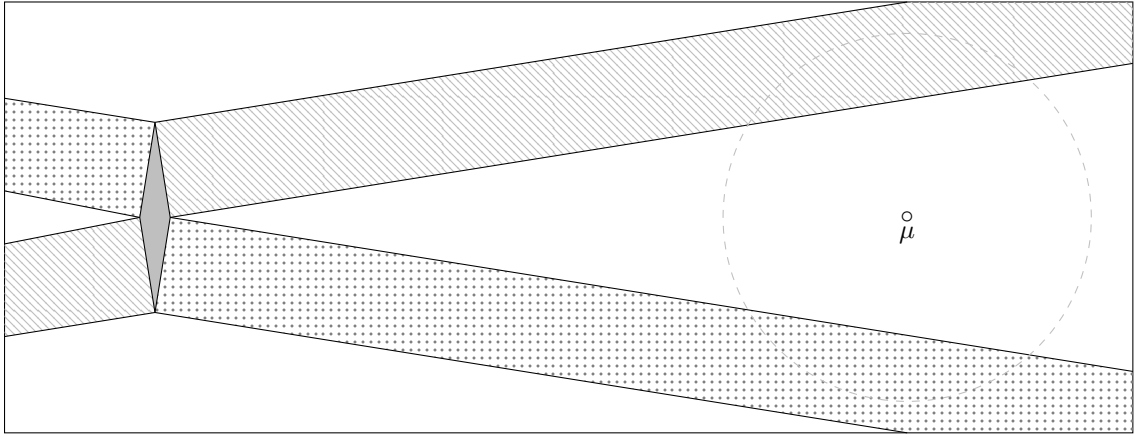
\begin{figure}[!ht]
\centering
\usetikzlibrary{patterns}
\begin{tikzpicture}[scale=1]
\draw (-2,- 2.866193) rectangle (13, 2.866193);

\begin{scope}
      \fill[lightgray] (0, 1.264911) -- ( 0.2025479, 0) -- (0,-1.264911) -- (-0.2025479, 0);
      \fill[pattern=north west lines, pattern color=lightgray] (0, 1.264911) -- (10,  2.866193) -- (13,2.866193) -- (13,2.049232) -- (0.2025479,  0);
      \fill[pattern=north west lines, pattern color=lightgray] (0, -1.264911) -- (-2,  -1.585167) -- (-2,-0.3526899) -- (-0.2025479,  0);
      \fill[pattern=dots, pattern color=gray] (0, 1.264911) -- (-2,  1.585167) -- (-2,0.3526899) -- (-0.2025479,  0);
      \fill[pattern=dots, pattern color=gray] (0, -1.264911) -- (10,  -2.866193) -- (13,-2.866193) -- (13,-2.049232) -- (0.2025479,  0);
  \end{scope}

\draw (0, 1.264911) -- (-0.2025479, 0);
\draw (0, 1.264911) -- ( 0.2025479, 0);
\draw (0,-1.264911) -- (-0.2025479, 0);
\draw (0,-1.264911) -- ( 0.2025479, 0);

\draw (0,  1.264911) -- (10,  2.866193);
\draw (0,  1.264911) -- (-2,  1.585167);
\draw (0, -1.264911) -- (10, -2.866193);
\draw (0, -1.264911) -- (-2, -1.585167);

\draw ( 0.2025479,  0) -- (13,  2.049232);
\draw ( 0.2025479,  0) -- (13, -2.049232);
\draw (-0.2025479,  0) -- (-2,  0.3526899);
\draw (-0.2025479,  0) -- (-2, -0.3526899);

\draw (10,0) node {$\circ$};
\draw (10,0) node[below] {$\mu$};
\draw[lightgray, dashed] (10,0) circle (2.447747);
\end{tikzpicture}
\caption{Plot of lasso selection regions in the $Y$ space: grey, dotted, diagonal lines and white regions indicate $S=\emptyset$, $S=\{1\}$, $S=\{2\}$ and $S=\{1,2\}$, respectively. $Y\sim N(\mu, I_2)$ and the dashed circle corresponds to the 
95\% quantile.}
\label{fig_tikz_example4}
\end{figure}

Figure \ref{fig_tikz_example4} shows the lasso selection regions in $Y$ space, e.g. the parallelogram region corresponds to $X_i^\mathsf{T}Y = \pm \lambda$ for $i=1,2$. The simulation setting was chosen to ensure that the probability of selecting no feature is almost zero, i.e., $\mathrm{P}(S=\emptyset) \approx 0$, and that the probability of selecting the first feature is equal to that of selecting the second feature, i.e., $\mathrm{P}(S=\{1\}) = \mathrm{P}(S=\{2\})$. Additionally, the conditional coverage is the same for both confidence intervals, i.e., $\mathrm{P}(\mathrm{P} \in C_{1\mid 1\in S}\mid S = \{1\}) = \mathrm{P}(\mathrm{P} \in C_{2\mid 2\in S}\mid S = \{2\})$ and $\mathrm{P}(\mathrm{P} \in C_{1\mid 1\in S}\mid S = \{1,2\}) = \mathrm{P}(\mathrm{P} \in C_{2\mid 2\in S}\mid S = \{1,2\})$. Table \ref{tab:MC_simulation} reports the estimated probabilities of selection by the lasso algorithm and the estimated conditional coverage of \citeauthor{liu2018}'s  $90\%$ confidence intervals based on a Monte Carlo simulation with $10^6$ repetitions.

\begin{table}[]
    \centering
    \begin{tabular}{|l|l|}
        \hline
    Probability of selection & Conditional coverage\\
    \hline
      $\mathrm{P}(S = \{i\})  = 0.0609065$  & $\mathrm{P}(\mathrm{P} \in C_{i\mid i\in S}\mid S = \{i\}) = 0.3878018 $  \\
      $\mathrm{P}(S = \{1,2\}) = 0.878187 $ & $\mathrm{P}(\mathrm{P} \in C_{i\mid i\in S}\mid S = \{1,2\}) = 0.935532$ \\ 
      \hline
    \end{tabular}
    \caption{The estimated probabilities of selection by the lasso algorithm and the estimated conditional coverage of \citeauthor{liu2018}'s  $90\%$ confidence intervals based on a Monte Carlo simulation with $10^6$ repetitions. The upper bound for Monte Carlo error is $\sqrt{0.25/10^{6}} = 0.0005$. }
    \label{tab:MC_simulation}
\end{table}

According to the table, when one feature is selected, the confidence intervals have substantially less coverage than the desired level of $90\%$. However, when both features are selected, the confidence intervals have a slightly higher coverage than the desired level. As expected, 
\begin{eqnarray*}
\mathrm{P}(\mathrm{P} \in C_{i \mid i \in S}\mid i \in S) &=& \mathrm{P}(\mathrm{P} \in C_{i \mid i \in S}\mid S = \{i\} \cup S=\{1,2\})\\
&=&\frac{\mathrm{P}(\mathrm{P} \in C_{i\mid i\in S}\mid S = \{i\}) \mathrm{P}(S = \{i\}) +  \mathrm{P}(\mathrm{P} \in C_{i\mid i\in S}\mid S = \{1,2\}) \mathrm{P}(S = \{1,2\})}{\mathrm{P}(S = \{i\}) + \mathrm{P}(S = \{1,2\})}\\
&=& 0.9
\end{eqnarray*}
for $i=1,2$, i.e. condition (\ref{eq liu}) holds.
However, the selective conditional confidence intervals of \cite{liu2018} do not control the unconditional FCR:
\begin{eqnarray*}
\mathrm{E}_P \bigg[ \frac{|\{i \in S\colon \mathrm{P} \notin C_{i\mid i \in S}\}|} {|S| \vee 1} \bigg] &=& 2\mathrm{P}(\mathrm{P} \notin C_{i\mid i\in S}\mid S = \{i\})\mathrm{P}(S = \{i\}) + \mathrm{P}(\mathrm{P} \notin C_{i\mid i\in S}\mid S = \{1,2\})\mathrm{P}(S = \{1,2\})  \\
&=& 0.1312 > 0.1
\end{eqnarray*}
Furthermore, adjusting the confidence level to $1-\alpha/|S|$ does not guarantee unconditional simultaneous control:
\begin{eqnarray*}
\mathrm{P}\Big(\mathrm{P} \in \bigcap_{i\in S} C_{i\mid  i \in S} \Big) &=& 2\mathrm{P}(\mathrm{P} \in C_{i\mid i\in S}\mid S = \{i\})\mathrm{P}(S = \{i\}) \\
&& +\ \mathrm{P}(\mathrm{P} \in \tilde{C}_{1\mid 1\in S} \cap \tilde{C}_{2\mid 2\in S} \mid S = \{1,2\})\mathrm{P}(S = \{1,2\})
=0.8777 < 0.9
\end{eqnarray*}
where $\tilde{C}_{i\mid i\in S}$ is the $i$th interval at the adjusted confidence level $(1-\alpha/2) = 95\%$ and  $\mathrm{P}(\mathrm{P} \in \tilde{C}_{1\mid 1\in S} \cap \tilde{C}_{2\mid 2\in S} \mid S = \{1,2\})=0.9456836$ is the estimated simultaneous conditional coverage when two features are selected.

\subsection{Prostate data set}

\cite{stamey1989prostate} was interested in the relation between prostate specific antigen (PSA) and several clinical measures, including log cancer volume (\texttt{lcavol}), log prostate weight (\texttt{lweight}), age, log of benign
prostatic hyperplasia amount (\texttt{lbph}), seminal vesicle invasion (\texttt{svi}), log of capsular penetration (\texttt{lcp}), the
Gleason score (\texttt{gleason}), and percent of Gleason scores 4 or 5 (\texttt{pgg45}). 
The dataset consisted of information collected from $n=97$ men who were preparing to undergo a radial prostatectomy.

The estimate from the regression model with all variables was utilized as the true value of $\sigma^2$. 
For this dataset, Liu et al. (2018) used $\lambda=0.0324$ (chosen by 10-fold cross-validation), which resulted in the selection of 7 variables. The following table compares the unconditional $p$-values for the hypotheses $\beta_i=0$ with the selective conditional $p$-values (the row corresponding to $\lambda=0.0324$) for the selected features (in black) and the non-selected features (in gray).

\begin{table}[ht]
\centering
\begin{tabular}{r|rrrrrrrr}
  \hline
$\lambda$ & lcavol & lweight & age & lbph & svi & lcp & gleason & pgg45 \\ 
  \hline
0.0324 & 0.0000 & 0.0036 & 0.0915 & 0.1779 & 0.0031 & 0.4666 & \textcolor{gray}{0.0073} & 0.7419 \\ 
0 & 0.0000 & 0.0020 & 0.0552 & 0.0949 & 0.0016 & 0.2380 & 0.7513 & 0.3072 \\ 
   \hline
\end{tabular}
\end{table}

In Figure \ref{fig pvalues}, the selective conditional  $p$-values are shown as a function of $\lambda$. It is worth noting that the $p$-values for the selected features begin with unconditional values and tend to either increase or increase and then decrease as they approach the critical threshold for selection. In contrast, the $p$-values for non-selected features start at 0 and eventually converge to the unconditional $p$-values.  

\begin{figure}[!ht]
    \centering
    \includegraphics[width=0.8\textwidth]{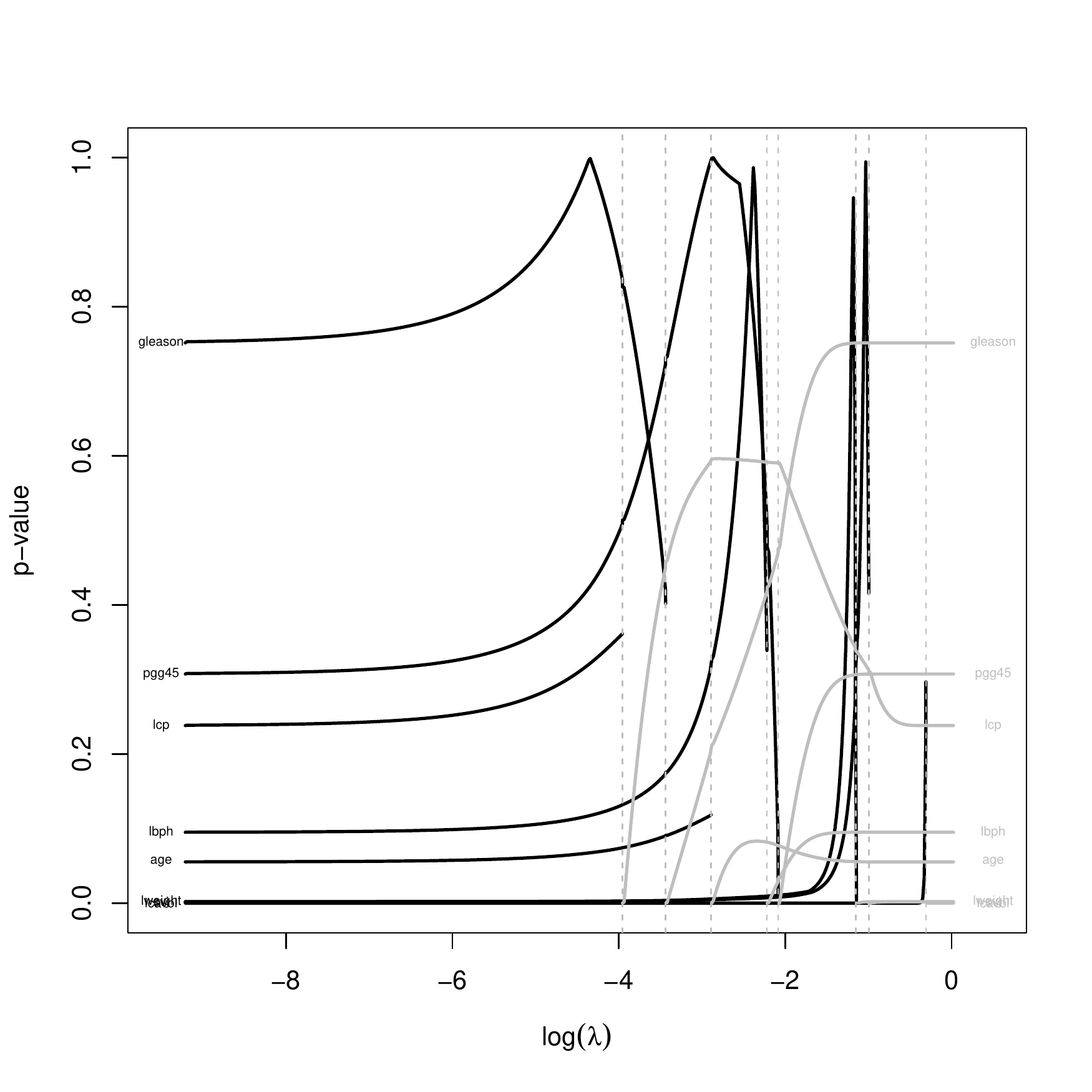}
    \caption{Selective conditional $p$-values as a function of $\lambda$. Black $p$-value curves are conditional on selection by the lasso, grey ones are
conditional on non-selection. Vertical dashed lines indicate the values of $\lambda$ at which the active set change.} 
    \label{fig pvalues}
\end{figure}

\bibliographystyle{chicago}

\end{document}